\documentclass[12pt]{article}
\usepackage{amsmath}
\usepackage{amssymb}
\usepackage{amsthm}
\usepackage{mathtools}
\usepackage{bm}
\usepackage{graphicx}
\graphicspath{{figures/}}
\usepackage{enumerate}
\usepackage{natbib}
\usepackage{url}
\usepackage{booktabs}
\usepackage{longtable}

\usepackage{xcolor}
\usepackage[normalem]{ulem}
\newcommand{\blind}{1}

\newcommand{\R}{\mathbb{R}}
\newcommand{\E}{\mathbb{E}}
\newcommand{\Var}{\operatorname{Var}}
\newcommand{\Prob}{\mathbb{P}}
\newcommand{\AUC}{\ensuremath{\mathrm{AUC}}}
\newcommand{\CL}{\ensuremath{\mathrm{CL}}}
\newcommand{\CI}{\ensuremath{\mathrm{CI}}}
\newcommand{\SEr}{\ensuremath{\mathrm{SE}}}
\newcommand{\Normal}{\mathcal{N}}
\newcommand{\GP}{\mathcal{GP}}
\newcommand{\Tcal}{\mathcal{T}}
\newcommand{\Dcal}{\mathcal{D}}
\DeclareMathOperator{\diag}{diag}

\DeclareMathOperator*{\argmax}{arg\,max}

\theoremstyle{plain}
\newtheorem{theorem}{Theorem}[section]

\newtheorem{corollary}[theorem]{Corollary}
\newtheorem{lemma}[theorem]{Lemma}
\theoremstyle{definition}

\theoremstyle{remark}

\begin{document}

\def\spacingset#1{\renewcommand{\baselinestretch}%
{#1}\small\normalsize} \spacingset{1}


\if1\blind
{
  \title{\bf Compartment-driven Gaussian processes for uncertainty quantification of pharmacokinetic exposures}
  \author{Sanjay Palta-Hill \\
    Department of Biomedical Data Science, Stanford University
    \and
    Barbara Engelhardt\thanks{Corresponding Author.}\\
    Department of Biomedical Data Science, Stanford University 
    \and 
    Edric Tam\thanks{Corresponding Author.} \\
    Department of Biomedical Data Science, Stanford University}
  \maketitle
} \fi

\if0\blind
{
  \bigskip
  \bigskip
  \bigskip
  \begin{center}
    {\LARGE\bf Compartment-driven Gaussian processes for uncertainty quantification of pharmacokinetic exposures}
\end{center}
  \medskip
} \fi

\vspace{-\baselineskip}
\begin{abstract}

Pharmacokinetic modeling is a crucial part of drug development. An important goal is to estimate patients' total and peak exposure to a drug over time, quantified by the area under the concentration--time curve (\AUC) and the maximum concentration ($C_{\max}$) respectively. Such exposure assessment is essential to balance efficacy and safety in medications that offer therapeutic benefits but carry risks of toxicity, such as chemotherapy drugs. Uncertainty quantification for exposure quantities, however, remains challenging, since they are functionals of an often sparsely observed concentration--time curve with measurement errors. We propose a compartment-driven Gaussian process framework that propagates uncertainty in concentration--time curves to pharmacokinetic exposure functionals. Our approach uses pharmacokinetic compartmental models to specify a mechanistically informed prior and Gaussian processes to accommodate departures from the assumed pharmacokinetic dynamics.  
The framework can accommodate both individual-level and group-aggregated concentration--time data commonly seen in practice. We provide straightforward computational procedures for drawing posterior inference on AUC and $C_{\max}$. Simulations across a range of observation schemes demonstrate robustness to departures from the assumed compartmental dynamics. We apply the method to patient-level datasets on warfarin and theophylline, and aggregated data from a chemotherapy study of 5-fluorouracil and capecitabine in patients with normal or impaired hepatic function. The chemotherapy analysis quantifies drug-specific exposure differences between groups with differing liver function, with potential implications for informing dosing decisions.
\end{abstract}

\noindent%

\newpage
\spacingset{1.9} 

\section{Introduction}
\label{sec:intro}

Pharmacokinetics (PK) is the quantitative study of how drug concentrations evolve in the human body after administration. The central object of study in PK is the plasma concentration of a drug, which evolves over time in response to processes related to the absorption, distribution, metabolism and excretion (ADME) of the drug \citep{rowland2011clinical}. The plasma concentration--time curve is frequently used to inform the therapeutic efficacy and toxicity of a drug. Quantitative characterization of the concentration--time curve forms the foundation of PK profiling, which is an explicit requirement for regulatory approval of new therapeutics \citep{gabrielsson2016pharmacokinetic}. Clinical decisions are also frequently made on the basis of summary quantities derived from the concentration--time curve, most notably the area under the curve (AUC), which measures total drug exposure over a time interval, and the maximum concentration ($C_{\max}$), which measures peak exposure. Other important pharmacokinetic quantities, such as clearance and bioavailability, may be derived from AUC estimates \citep{rowland2011clinical}.

These PK exposure quantities are particularly important to profile for drugs that have substantial toxicity, where one has to select dosages that trade off between therapeutic benefit and unwanted side effects \citep{holford1981understanding}. For example, in oncology, dosing of certain chemotherapy agents is often individualized to a target AUC level with the goal of maintaining therapeutic exposure while limiting severe, exposure-related toxicity  \citep{calvert1989carboplatin}. In psychiatry, mood stabilizers and antiepileptics such as lithium and valproate have concentration-dependent toxicity, making monitoring of plasma concentrations a routine clinical practice \citep{hiemke2018consensus}. In transplantation medicine, immunosuppressants must be maintained within a narrow concentration range to balance graft rejection against infection \citep{dejonge2005therapeutic}.

Despite the central role of AUC and $C_{\max}$ in pharmacokinetics and clinical decision making, uncertainty quantification for the AUC and $C_{\max}$ remains challenging because they are functionals of a concentration--time curve that is reconstructed from sparse, noisy plasma measurements. It has been documented that uncertainty reporting for AUC and $C_{\max}$ is often incomplete, with many studies presenting these quantities as point estimates without confidence or credible intervals \citep{van2011quality}.

The standard approach to estimating concentration--time curves uses parametric compartmental models, in which a system of differential equations encodes the absorption, distribution, and elimination of a drug. Inference under these models generally does not account for departures from the assumed compartmental dynamics, which necessarily simplify the underlying biological processes. An alternative approach is noncompartmental analysis (NCA) \citep{gabrielsson2016pharmacokinetic}, which avoids specifying a complete dynamic model. Instead, it estimates the AUC by trapezoidal integration directly on the observed concentration--time data and takes $C_{\max}$ to be the largest observed concentration. These estimators alone do not directly provide uncertainty quantification for an individual’s exposure given sparse, error-prone measurements. 
 
Resampling approaches do not automatically resolve these limitations. Since the maximum is a non-smooth functional, the ordinary bootstrap can fail in this setting and requires additional conditions for validity \citep{bickel1981asymptotic,fang2019inference}. Standard nonlinear mixed-effects models \citep{sheiner1980evaluation,bauer2019nonmem} borrow information across patients by allowing pharmacokinetic parameters to vary between individuals. However, the concentration–time curves remain constrained to a prescribed parametric family. Off-the-shelf Gaussian process (GP) approaches \citep{rasmussen2006gaussian} are flexible but do not exploit known pharmacological structure nor prior knowledge encoded in compartmental models.
 
Motivated by these challenges, we propose a compartment-driven Gaussian process (CD-GP) framework for inference on functionals of PK concentration--time curves, combining the mechanistic structure of compartmental pharmacokinetic models with the flexibility of Gaussian processes. We regard uncertainty in exposure summaries as arising primarily from uncertainty in the underlying concentration--time curve, which must be reconstructed from sparse and error-prone observations. Once a concentration--time curve is fixed, quantities such as AUC and $C_{\max}$ are fully specified. Consequently, the posterior distribution of the curve induces posterior distributions for these functionals, providing a principled mechanism for uncertainty propagation. This shared source of uncertainty allows us to infer multiple exposure functionals from a single model of the curve. Moreover, since our framework treats the concentration--time curve as the primary object to model, it does not depend on whether the sampled observations come from a single patient (in which case the inferential target is that patient's exposure) or from an aggregated cohort summary (in which case the inferential target is the exposure of the population-mean curve). This allows for a unified approach to handle both individual-level and group-aggregated data commonly encountered in pharmacokinetic studies.

The CD-GP models the concentration--time curve as a GP, with the prior mean specified as the solution of a compartmental pharmacokinetic model, so that established pharmacological knowledge informs the analysis. The GP component captures residual deviations from the compartmental model and allows the posterior to depart from the parametric form when warranted by the data. The resulting posterior is a distribution over plausible concentration--time curves. Posterior distributions for any exposure functional, including AUC, $C_{\max}$, clearance, and other downstream quantities, are obtained by applying the corresponding functional to draws from this posterior. We show that this yields a closed-form posterior distribution for AUC, owing to the closure of Gaussian processes under integration, and enables consistent posterior approximation for $C_{\max}$. The framework retains the interpretability of mechanistically informed compartmental models, while adding the flexibility and principled uncertainty quantification of a GP component.

We apply our model to individual patient time series for warfarin and theophylline, two commonly used drugs. We also illustrate our model on aggregated concentration--time summaries for capecitabine and 5-fluorouracil (5-FU), chemotherapy agents used in colorectal and other cancers \citep{twelves1999hepatic}. Because liver dysfunction may alter the metabolism and elimination of capecitabine and 5-FU \citep{reigner2001clinical}, an important clinical question is whether liver dysfunction changes drug exposure sufficiently to warrant dose adjustment. Our framework enables probabilistic comparisons of drug exposure by converting the concentration--time summaries commonly reported in the pharmacokinetic literature into posterior distributions for AUC.

Our contribution in this paper is a CD-GP framework that treats incomplete knowledge of the concentration--time curve as the common source of uncertainty in different PK exposures. We combine prior pharmacological knowledge encoded in compartmental PK models with a flexible GP discrepancy that allows the concentration--time curve to depart from the assumed compartmental dynamics. By modeling the concentration--time curve directly and propagating uncertainty to its functionals, we obtain a unified Bayesian approach to inference on exposure quantities such as AUC and $C_{\max}$. The framework accommodates individual-level and group-aggregated observations and can be extended to incorporate more sophisticated mechanistic priors and other exposure functionals.

The remainder of the paper is organized as follows. Section~\ref{sec:related} reviews related work. Section~\ref{sec:prelim} establishes notation and provides the relevant background on compartmental pharmacokinetic models, GPs, and exposure functionals. Section~\ref{sec:model} describes the proposed model and develops posterior inference for both linear and nonlinear exposure functionals, with supporting theoretical results stated alongside each inferential procedure and proofs given in the Supplementary Material. Section~\ref{sec:sim} reports the simulation study. Section~\ref{sec:case} reports case studies on individualized data (warfarin) and on aggregated cohort-mean data (5-FU and capecitabine); a further individualized analysis of theophylline is given in Section~\ref{sec:theo} of the Supplementary Material. Section~\ref{sec:disc} discusses extensions and limitations.

\section{Related work}
\label{sec:related}

Compartmental pharmacokinetic models trace back at least to \citet{bateman1910solution} and were systematized in the latter half of the twentieth century into the one- and two-compartment forms now standard in pharmacology textbooks \citep{wagner1975fundamentals, gibaldi1982pharmacokinetics, gabrielsson2016pharmacokinetic}. These models represent the human body as a small number of well-mixed compartments and yield analytic concentration--time curves using systems of differential equations. Noncompartmental analysis (NCA), originating with \citet{yeh1978comparison}, avoids imposing structural assumptions by directly computing the AUC from the observed data via the trapezoidal rule and reporting the largest observed concentration as $C_{\max}$. Both approaches provide point estimates of exposures with limited tools for uncertainty quantification.

Population pharmacokinetics extends compartmental modeling to hierarchical analyses of multi-patient datasets, often through the nonlinear mixed-effects framework. \citet{sheiner1980evaluation, sheiner1981evaluation} introduced the NONMEM software for estimation of population pharmacokinetic parameters. \citet{lindstrom1990nonlinear} and \citet{pinheiro2000mixed} developed the broader statistical theory and software for nonlinear mixed-effects models. Bayesian formulations were developed by \citet{wakefield1996bayesian} and \citet{wakefield1994bayesian}, and surveyed by \citet{davidian1995nonlinear, davidian2003nonlinear}. Stochastic-approximation expectation--maximization algorithms \citep{delyon1999convergence, lavielle2014mixed} have made population pharmacokinetic estimation routine in large datasets. Software implementations, including Monolix and the \texttt{nlme} package in R, are now widely used in clinical pharmacology. 

The use of Gaussian processes (GPs) for regression and functional modeling has a long history in statistics and machine learning \citep{ohagan1978curve, rasmussen2006gaussian}. GP methods have been applied to longitudinal biomedical data in settings such as disease-trajectory modeling \citep{schulam2015framework} and clinical time-series prediction \citep{futoma2017improved}. Bayesian frameworks for calibrating computer models \citep{kennedy2001bayesian,higdon2008computer, brynjarsdottir2014learning} use physical observations to infer unknown simulator parameters while accounting for systematic discrepancies between model predictions and the underlying physical processes. 

Direct application of GP regression to pharmacokinetic concentration--time data remains comparatively rare. Closest in spirit to our construction is \citet{wheeler2014mechanistic}, who built hierarchical GP priors that favor curves consistent with a linear ordinary differential equation system that describes muscle mechanics. Our work differs in its inferential goal. Rather than estimating mechanistic parameters or borrowing strength across curves, we take the concentration--time curve itself as the object of inference and derive posterior distributions for the exposure functionals of clinical interest. 

Inference for functionals of GPs is a crucial part of our proposed framework. Linear functionals of GPs remain Gaussian and admit explicit expressions for their means and variances, a fact used in Bayesian quadrature \citep{ohagan1991bayes, briol2019probabilistic} and in the broader theory of linear operators on GP priors \citep{sarkka2011linear}. Inference for nonlinear functionals, in particular the supremum, has been studied in the theoretical literature on the geometry of Gaussian random fields \citep{adler2007random} and concentration inequalities for Lipschitz functionals of Gaussian vectors \citep{borell1975brunn, sudakov1974extremal, boucheron2013concentration}. These tools to our knowledge have not been applied to posterior inference for $C_{\max}$ in pharmacokinetics.

\section{Preliminaries}
\label{sec:prelim}

This section establishes the notation used throughout the paper, introduces the compartmental pharmacokinetic models that will serve as the mechanistically informed prior in our framework, and reviews Gaussian processes together with the relevant covariance functions. 

\subsection{Setup and notation}
\label{sec:notation}

The primary inferential object throughout is the concentration--time curve $C:\Tcal\rightarrow\R,$ where $\Tcal=[a,b]$. 
Exposure summaries such as AUC and $C_{\max}$ are obtained as functionals of $C$. Depending on the available data, $C$ represents either the curve of an individual patient or the mean curve of a clinically defined group.

When individual-level measurements are available for patient $i$, indexed by $i = 1, \ldots, N$, we write
$
T_i=(t_{i1},\ldots,t_{in_i})^\top
$ and $
Y_i=(Y_{i1},\ldots,Y_{in_i})^\top,
$
where $T_i$ contains the observation times and $Y_i$ contains the corresponding plasma concentrations. The data are
$
\Dcal_i=\{(t_{ij},Y_{ij}):j=1,\ldots,n_i\},
$
and the observation model is
$
Y_i=C_i(T_i)+\epsilon_i$, where $
\epsilon_i\sim\Normal(0,R_i)$. The inferential target is the individual curve $C_i$ and its exposure functionals.

In pharmacokinetic studies, concentrations may be reported only as averages over patients within a clinically defined group. Such data take the form
$
\overline{\Dcal}
=
\{(t_j,\overline Y_j,\SEr_j):j=1,\ldots,n\},
$
where $\overline Y_j$ is the reported mean concentration at time $t_j$ and $\SEr_j$ is its standard error. The target is then the group-mean curve
$
\overline C(t)=\E_i\{C_i(t)\},
$
where the expectation is over patients represented by that clinical group. Writing
$
T=(t_1,\ldots,t_n)^\top$, $
\overline Y=(\overline Y_1,\ldots,\overline Y_n)^\top,
$
we use the observation model
$
\overline Y=\overline C(T)+\overline\epsilon$, where $
\overline\epsilon\sim\Normal(0,R).
$
When only marginal standard errors are reported, we take
$
R=\diag(\SEr_1^2,\ldots,\SEr_n^2).
$

The two data sources target different exposure quantities. By linearity,
\[
\int_a^b\overline C(t)\,dt
=
\E_i\left\{\int_a^b C_i(t)\,dt\right\},
\]
so the AUC of the group-mean curve equals the mean of the individual AUCs. In contrast,
$
\sup_{t\in\Tcal}\overline C(t)
\leq
\E_i\left\{\sup_{t\in\Tcal}C_i(t)\right\},
$
so the peak of the group-mean curve is generally not the mean of the individual peaks. Accordingly, posterior intervals obtained from the two data sources have different clinical interpretations. The results below apply to either setting by taking $C$ to denote the relevant target curve.

We write $\Normal(\mu,\Sigma)$ for a multivariate normal distribution, and $\GP(m,k)$ for a Gaussian process with mean function $m$ and covariance kernel $k$. The standard normal cumulative distribution function is denoted by $\Phi$, and $I_d$ denotes the $d\times d$ identity matrix. For a function $f$ and evaluation points $T=(t_1,\ldots,t_n)^\top$, we write
$
f(T)=(f(t_1),\ldots,f(t_n))^\top.
$
For $T\in\R^n$, $U\in\R^q$, and a kernel function $k$, let $k(T,U)\in\R^{n\times q}$ have $(r,s)$th entry $k(t_r,u_s)$, and write $K_{TT}=k(T,T)$.

\subsection{Compartment models in pharmacokinetics}
\label{sec:cpt}

Compartmental pharmacokinetic models represent the body as a small number of well-mixed kinetic compartments connected by first-order transfer rates \citep{gibaldi1982pharmacokinetics, rowland2011clinical}. Each compartment corresponds to a physiological space where the drug can reside. For example, blood plasma and tissues receiving high blood flow are often represented collectively as the so-called central compartment. Transfers between compartments and elimination from the system are governed by rate constants, and the dynamics form a system of linear ordinary differential equations. We write $m_i(t; \theta_i)$ for the plasma concentration predicted by the model at time $t$ for patient $i$. Here, $\theta_i$ collects the pharmacokinetic parameters of patient $i$. Compartmental PK models can be formulated in terms of either the total amount of drug in each compartment or the drug concentration within each compartment. The concentration is obtained by dividing the drug amount by the corresponding compartment's apparent volume of distribution \citep{rowland2011clinical}. We use the concentration formulation throughout because the concentration--time curve is our primary object of interest. Accordingly, $c_{0,i}$ denotes the initial dose normalized by the drug's so-called apparent volume of distribution.
 
The simplest PK compartment model describes a single dose of drug injected rapidly into the bloodstream, a form of administration known as an intravenous bolus \citep{rowland2011clinical}. In the one-compartment intravenous bolus model, the dose is delivered to the central compartment instantaneously and is thereafter eliminated by a first-order process at the elimination rate constant $\kappa_{e,i}$:
\[
\frac{d\, m_i(t; \theta_i)}{dt} = -\kappa_{e,i}\, m_i(t; \theta_i), \qquad m_i(0; \theta_i) = c_{0,i}.
\]
This gives the plasma concentration $m_i(t; \theta_i) = c_{0,i} \exp(-\kappa_{e,i} t)$, with parameter vector $\theta_i = (\kappa_{e,i}, c_{0,i})$.

For a drug that is administered orally, the simplest one-compartment model assumes absorption from a so-called depot, which represents drug remaining in the gut that has not yet been absorbed into the bloodstream. Drug passes from the depot into the central compartment by a first-order process with absorption rate constant $\kappa_{a,i}$. Writing $\xi_i(t)$ for the drug concentration in the depot, we have:
\[
\frac{d\, \xi_i(t)}{dt} = -\kappa_{a,i}\, \xi_i(t), \qquad
\frac{d\, m_i(t; \theta_i)}{dt} = \kappa_{a,i}\, \xi_i(t) - \kappa_{e,i}\, m_i(t; \theta_i),
\]
with initial conditions $\xi_i(0) = c_{0,i}$ and $m_i(0; \theta_i) = 0$. For $\kappa_{a,i}\ne\kappa_{e,i}$, solving the system gives the plasma concentration \citep{bateman1910solution, gibaldi1982pharmacokinetics}
\begin{equation}
m_i(t; \theta_i) = \frac{c_{0,i}\, \kappa_{a,i}}{\kappa_{a,i} - \kappa_{e,i}} \left\{e^{-\kappa_{e,i} t} - e^{-\kappa_{a,i} t}\right\},
\label{eq:onecpt}
\end{equation}
where $\theta_i = (\kappa_{a,i}, \kappa_{e,i}, c_{0,i})$ collects the absorption rate constant, the elimination rate constant, and the initial dose on the concentration scale. 
When $\kappa_{a,i}=\kappa_{e,i}=\kappa_i$, the continuous limiting expression is $m_i(t;\theta_i)=c_{0,i}\kappa_i t e^{-\kappa_i t}$.

More complex compartment models can be built by adding compartments or by replacing the linear elimination term with a nonlinear one. For example, the two-compartment model couples the central compartment to a peripheral tissue compartment. The exchange of drug between the compartments is described by the transfer rate constants $\kappa_{12}$ and $\kappa_{21}$, with elimination occurring only from the central compartment. The resulting plasma concentration is a sum of two decaying exponentials, $c_1\, e^{-\lambda_1 t} + c_2\, e^{-\lambda_2 t}$, where $c_1, c_2, \lambda_1, \lambda_2$ are explicit functions of the rate constants and the dose \citep{gibaldi1982pharmacokinetics}. 

Linear compartmental models with constant coefficients often admit closed-form expressions that are available in standard pharmacokinetic references \citep{wagner1975fundamentals, gibaldi1982pharmacokinetics, gabrielsson2016pharmacokinetic} and admit straightforward extension to multi-dose settings. Closed-form solutions for more complex nonlinear compartmental models, such as those incorporating Michaelis--Menten elimination \citep{gabrielsson2016pharmacokinetic}, are not in general available, and the corresponding initial-value problem is solved numerically, typically via the Runge--Kutta method. 

\subsection{Pharmacokinetic functionals}
\label{sec:functionals}

The main inferential targets of interest in this work can be stated as functionals of the concentration--time curve $C$. The total area under the concentration--time curve over the observation window is given by the linear functional
\begin{equation}
\AUC_{a:b} = \int_a^b C(t)\, dt.
\label{eq:auc}
\end{equation} On the other hand, the maximum concentration over $\Tcal$, written $C_{\max}$ throughout this paper, is
\begin{equation}
C_{\max} = \sup_{t \in \Tcal} C(t),
\label{eq:cmax}
\end{equation} which is a nonlinear functional of the curve $C$. Where a particular patient is referred to, we add a subscript $i$, as in $C_{\max,i}$.

A number of other clinically relevant pharmacokinetic quantities are built upon the AUC functional. An important example is the systemic clearance of a drug for an intravenous dose $D$, which is given by 
\begin{equation}
\CL = \frac{D}{\AUC_{0:\infty}}.
\label{eq:clearance}
\end{equation}
In practice, $\AUC_{0:\infty}$ may be approximated by $\AUC_{0:b}$ over a finite time interval when the unobserved tail contribution is negligible. Other examples include the bioavailability and the elimination half-life of a drug. Because these quantities can be written as functionals of the concentration--time curve, if a posterior distribution over the concentration--time curve is available, the posterior distribution of each PK quantity can be obtained by evaluating its corresponding functional on posterior draws of the concentration--time curve.

\subsection{Gaussian processes}
\label{sec:gp}

A Gaussian process (GP) is a collection of random variables $\{g(t) : t \in \Tcal\}$ such that, for every finite set of evaluation points $t_1, \ldots, t_n \in \Tcal$, the vector $(g(t_1), \ldots, g(t_n))^\top$ is multivariate normal \citep{rasmussen2006gaussian}. It is fully specified by a mean function $m : \Tcal \to \R$ and a symmetric positive semidefinite covariance kernel $k : \Tcal \times \Tcal \to \R$, and we write $g \sim \GP(m, k)$. When a GP is observed under additive Gaussian noise, the conditional distribution of the process given the observations is again a GP with closed-form mean and covariance functions \citep{rasmussen2006gaussian}; we defer the explicit posterior formulas to Section~\ref{sec:postcurve}, where they appear in the context of the proposed model.

In this work we focus on stationary covariance kernels of the form $k_\phi(s, t) = \sigma_f^2\, \rho_\nu(|t - s| / \ell)$, parameterized by a marginal variance $\sigma_f^2 > 0$ and a length-scale $\ell > 0$, collected as $\phi = (\sigma_f^2, \ell)$. The radial profile $\rho_\nu$ is taken from the Mat\'ern family, indexed by a smoothness parameter $\nu$: the process is $\lceil \nu \rceil - 1$ times mean-square differentiable, so that larger $\nu$ yields smoother realizations, with the squared-exponential kernel arising in the limit $\nu \to \infty$. Mat\'ern kernels are widely used in spatial statistics and GP regression because they permit explicit control over prior regularity \citep{stein1999interpolation,guttorp2006studies,rasmussen2006gaussian}. Kernel hyperparameters may be fixed using prior knowledge, selected by cross-validation, estimated by maximizing the marginal likelihood, or assigned prior distributions and inferred jointly with the latent process \citep[Chapter~5]{rasmussen2006gaussian}.

\section{Compartment-driven Gaussian process model}
\label{sec:model}

We now introduce our proposed compartment-driven Gaussian process model and develop posterior inference for the pharmacokinetic exposure functionals defined in Section~\ref{sec:functionals}. The goal is to produce a posterior distribution over the concentration--time curve $C$ and, by direct propagation, obtain posterior distributions over functionals such as $\AUC_{a:b}$ and $C_{\max}$ given the observed data $\Dcal$. Theoretical results are stated alongside the inferential procedures, whereas proofs are deferred to the Supplementary Material.

\subsection{Compartment-driven Gaussian process (CD-GP)}
\label{sec:hyper}

Let $C$ denote the target concentration--time curve of Section~\ref{sec:notation}, observed at times $T$ with concentrations $Y$. Conditional on a hyperparameter vector $\psi = (\theta, \phi, \sigma_\epsilon^2)$, the proposed model is
\begin{equation}
\begin{aligned}
Y \mid C, \psi &\sim \Normal\!\left(C(T),\, R\right), \\
C(t) &= m(t;\theta) + g(t), \\
g \mid \phi &\sim \GP(0, k_\phi),
\end{aligned}
\label{eq:model}
\end{equation}
or equivalently $C \mid \psi \sim \GP(m, k_\phi)$. The prior mean function $m$ is the solution of a compartmental pharmacokinetic model as described in Section~\ref{sec:cpt}, where the practitioner can choose a pharmacokinetic model appropriate for the drug under study. The one-compartment model with oral absorption (Equation~\ref{eq:onecpt}) is our default. When patient covariates $x$ are available, they may be incorporated into the specification of $m$. The covariance kernel $k_\phi$ is a stationary kernel introduced in Section~\ref{sec:gp}.

The two data types of Section~\ref{sec:notation} are recovered as special cases. For individual patient data, we can take $C$ to be the patient's curve $C_i$ and $R$ to be the observation covariance $\sigma_\epsilon^2 I_{n_i}$. For aggregated clinical data, we can take $C$ to be the mean curve $\overline C$ over the clinical cohort and $R$ to be $\diag(\SEr_1^2,\ldots,\SEr_n^2)$. Everything that follows is stated for generic $C$ and $R$, and therefore applies to either setting without modification.

The model is motivated by two requirements. First, established knowledge encoded in compartmental and population pharmacokinetic models should inform the prior, which helps to stabilize inference under sparse observations and anchor estimated curves to pharmacologically plausible dynamics. Second, compartmental models are necessarily simplifications of the underlying biology. We therefore use a GP residual to capture systematic departures from the assumed compartmental dynamics. The additive decomposition of $C$ into a mechanistic mean function $m$ and a GP residual $g$ gives a simple approach to achieve both objectives. Here the mean $m$ governs interpolations and extrapolations across time, whereas the kernel $k_\phi$ governs the size and smoothness of deviations from the mean.

The Gaussian specification (Equation~\ref{eq:model}) does not impose nonnegativity on either the latent curve or the observations. For most applications focused on AUC and $C_{\max}$, this Gaussian formulation can be appropriate when concentrations are appreciably above zero relative to their uncertainty over the dosing interval of interest. When substantial probability is assigned to negative concentrations over an appreciable part of the interval, a log transformation or another positivity-preserving formulation may be more appropriate.

We use the Mat\'ern-5/2 kernel throughout, giving a twice mean-square differentiable process. 


\subsection{Posterior concentration--time curve}
\label{sec:postcurve}

All exact posterior statements below are conditional on fixed hyperparameters. Let $K_{TT} = k(T, T) \in \R^{n \times n}$, let $\Sigma_Y = K_{TT} + R$ denote the marginal covariance of the observations, and assume that $\Sigma_Y$ is positive definite. Set $m(T) = (m(t_1), \ldots, m(t_n))^\top$ and $k(t, T) = (k(t, t_1), \ldots, k(t, t_n))$.

The posterior distribution $C \mid Y $ of the concentration--time curve is the object from which the posterior distributions of all exposure functionals are derived. We give a characterization of $C \mid Y $ below. 

\begin{theorem}[Posterior Gaussian process]
\label{thm:gp-posterior}
Under the model in~Equation~(\ref{eq:model}), the posterior distribution over the concentration--time curve is given by the GP
\begin{equation}
C \mid Y = y \sim \GP(\mu, K^+),
\label{eq:gp-posterior}
\end{equation}
with posterior mean and covariance
\begin{equation}
\mu(t) = m(t) + k(t, T)\, \Sigma_Y^{-1}\, \{y - m(T)\},
\label{eq:postmean}
\end{equation}
\begin{equation}
K^+(s, t) = k(s, t) - k(s, T)\, \Sigma_Y^{-1}\, k(T, t).
\label{eq:postcov}
\end{equation} 
\end{theorem}

In practice, $\Sigma_Y^{-1}$ need not be formed explicitly. A common computational trick in the GP literature is to compute a Cholesky factorization of $\Sigma_Y$ once, after which the linear systems required for the posterior mean and covariance are evaluated efficiently by forward and backward substitution.

\subsection{Posterior inference for linear exposure functionals}

The posterior $C \mid Y = y$ (Equation~\ref{eq:gp-posterior}) characterizes uncertainty in the entire concentration--time curve. This uncertainty propagates simply to linear functionals. In particular, since $\AUC_{a:b}=\int_a^b C(t)\,dt$ is the integral of the concentration--time curve over the exposure window, it is a linear functional of $C$ and its posterior is exactly Gaussian. This is elaborated below.

\begin{theorem}[Gaussianity of AUC]
\label{thm:auc-gaussian}
Under the posterior in Equation~\ref{eq:gp-posterior}, assume that the mean and covariance functions are continuous and that $C$ has almost surely continuous sample paths on the compact interval $[a,b]$, with $a<b$. Then
\[
\AUC_{a:b}\mid Y=y\sim\Normal(\mu_A,\sigma_A^2),
\]
where
\[
\mu_A=\int_a^b\mu(t)\,dt,
\qquad
\sigma_A^2=\int_a^b\int_a^b K^+(s,t)\,ds\,dt.
\]
\end{theorem}

The corresponding $(1-\alpha)$-level credible interval for $\AUC_{a:b}$ is
\[
[\mu_A-z_{1-\alpha/2}\sigma_A,\,
 \mu_A+z_{1-\alpha/2}\sigma_A],
\]
where $z_{1-\alpha/2}$ is the standard normal quantile.

We evaluate the integrals defining the AUC posterior mean and variance using deterministic Gaussian quadrature, which enjoys strong guarantees on consistency. Table~\ref{tab:algorithm} provides a summary of the computation, whereas Section~\ref{sec:auc-numerics} of the Supplementary Material provides background on quadrature and our related consistency results. The practical implication of these results is that the closed-form Gaussian posterior of $\AUC_{a:b}$ can be evaluated to arbitrary precision numerically. Posterior means and credible intervals then follow directly from the resulting univariate normal distribution, without the need for any Monte Carlo simulation. 

\subsection{Posterior inference for nonlinear exposure functionals}

The posterior of $C_{\max}$ is generally not available in closed form because it is a nonlinear functional of the GP. We approximate this posterior by drawing concentration vectors on a dense grid from the posterior Gaussian distribution and taking the maximum of each draw. Posterior summaries and credible intervals are obtained from these simulated maxima. Table~\ref{tab:algorithm} summarizes the computational procedures, and Section~\ref{sec:S-max-functionals} of the Supplementary Material establishes convergence of the grid approximation.

\begin{figure}[!htbp]
\begingroup
\spacingset{1}
\centering
\includegraphics[width=\linewidth]{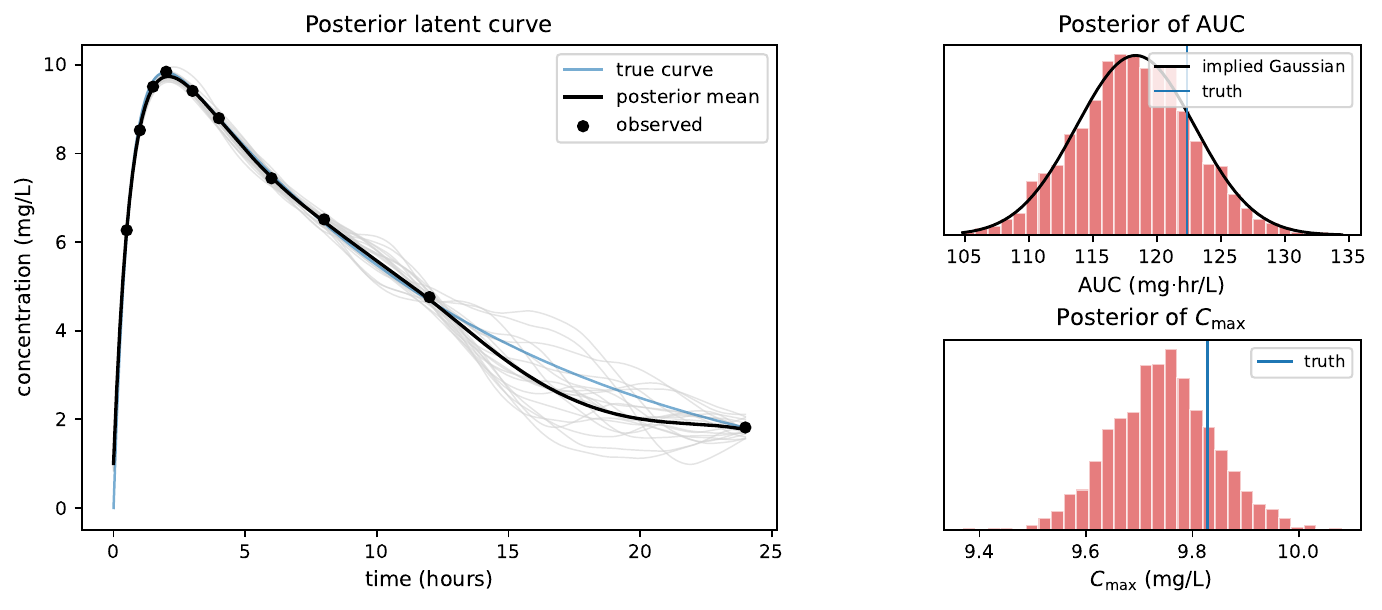}

\caption{Posterior path samples for a single target curve. Left: twenty posterior paths (light gray) and posterior mean (solid black), with observed data points overlaid. Right top: histogram of the posterior of $\AUC_{a:b}$, overlaid with the Gaussian density implied by the closed-form posterior mean and variance. Right bottom: histogram of the posterior of $C_{\max}$.
\label{fig:paths}}
\endgroup
\end{figure}

The posterior maximum satisfies a two-sided concentration bound around its posterior expectation.

\begin{theorem}[Posterior concentration of $C_{\max}$]
\label{thm:cmax-conc}
Suppose $C\mid Y=y\sim\GP(\mu,K^+)$ has continuous mean and covariance functions and almost surely continuous sample paths on $[a,b]$. Let
\[
m_\star=\E(C_{\max}\mid Y=y)<\infty,
\qquad
\sigma_\star^2=\sup_{t\in[a,b]}K^+(t,t)>0.
\]
Then, for every $r>0$,
\[
\Prob\!\left(|C_{\max}-m_\star|\ge r\mid Y=y\right)
\le 2\exp\!\left(-\frac{r^2}{2\sigma_\star^2}\right).
\]
Consequently, for every $\alpha\in(0,1)$,
\[
\Prob\!\left(
|C_{\max}-m_\star|\le\sigma_\star\sqrt{2\log(2/\alpha)}
\,\middle|\,Y=y\right)\ge1-\alpha.
\]
\end{theorem}

Here $\sigma_\star$ is the largest posterior pointwise standard deviation. The bound quantifies uncertainty in peak drug concentration. The bound depends on the largest posterior pointwise standard deviation over the full interval. The proof and supporting results are provided in Section~\ref{sec:S-conc-ineq} of the Supplementary Material.

Table~\ref{tab:algorithm} summarizes posterior inference for individual curves and group-mean curves. The two settings are distinguished by the observation covariance $R$. 

\begin{table}[!htbp]
\begingroup
\spacingset{1}\small
\centering
\caption{Posterior inference for AUC and $C_{\max}$.\label{tab:algorithm}}
\begin{tabular}{@{}p{\linewidth}@{}}
\toprule
\textbf{Input:} data $(T, Y)$; mechanistic mean $m(\cdot)$; kernel $k_\phi$; observation covariance $R$; quadrature/grid nodes $U_q$ and weights $\omega_q$; number of posterior paths $S$; jitter $\eta$. \\
\midrule
\textbf{Step 1.} Compute $K_{TT} = k_\phi(T, T)$ and $\Sigma_Y = K_{TT} + R$. \\
\textbf{Step 2.} Compute Cholesky factor $L_\Sigma$ of $\Sigma_Y$ and solve $L_\Sigma \beta_1 = Y - m(T)$, $L_\Sigma^\top \beta = \beta_1$. \\
\textbf{Step 3.} Posterior mean on $U_q$: $\mu(U_q) = m(U_q) + k_\phi(U_q, T)\, \beta$. \\
\textbf{Step 4.} Posterior covariance on $U_q$: $K^+(U_q, U_q) = k_\phi(U_q, U_q) - k_\phi(U_q, T)\, \Sigma_Y^{-1}\, k_\phi(T, U_q)$, computed via Cholesky solves. \\
\textbf{Step 5.} Quadrature approximation to the AUC posterior: $\mu_{A,q}=\omega_q^\top\mu(U_q)$, $\sigma_{A,q}^2=\omega_q^\top K^+(U_q,U_q)\omega_q$. \\
\textbf{Step 6.} Posterior paths: factor $K^+(U_q,U_q) + \eta I_q = L_K L_K^\top$. For $s = 1, \ldots, S$, draw $Z_s \sim \Normal(0, I_q)$ and set $C^{(s)}(U_q) = \mu(U_q) + L_K\, Z_s$. \\
\textbf{Step 7.} Path summaries: $A^{(s)} = \omega_q^\top C^{(s)}(U_q)$, $M^{(s)} = \max_v C^{(s)}(u_v)$. \\
\textbf{Step 8.} Report posterior summaries: mean, median, and $(1-\alpha)$-credible intervals for $\AUC_{a:b}$ and $C_{\max}$. \\
\bottomrule
\end{tabular}
\endgroup
\end{table}

\subsection{Practical considerations}
\label{sec:practical}

We discuss several implementation considerations for our compartment-driven GP model. 

\textit{Hyperparameter selection.} The kernel hyperparameters $(\sigma_f, \ell)$ and, for individualized data, the observation-noise scale $\sigma_\epsilon$ can be chosen in different ways. A common approach is type-II maximum likelihood, or empirical Bayes, which maximizes the marginal likelihood of the observations with the latent process integrated out. This can be convenient here because that marginal likelihood is available in closed form \citep{rasmussen2006gaussian}. When few observations are available, often the case in pharmacology applications, unpenalized likelihood-based estimates are prone to degeneracy. For example, variance parameters might collapse towards zero and length scale parameters can drift beyond the observation window. A common remedy is to penalize the likelihood with weakly-informative priors, applied either through a maximum a posteriori estimate or through a fully Bayesian treatment that integrates the hyperparameters out \citep{gelman2006prior, simpson2017penalising, fuglstad2019constructing}. A further option is to fix a hyperparameter outright at a value implied by subject-matter knowledge. This is often the pragmatic choice in pharmacokinetic applications, where there is usually ample prior information on the relevant quantities for commonly used drugs.  
Any of these approaches may be reasonable depending on the length of the observed trajectories, the sampling design, and the functional of interest. We give detailed documentation of the hyperparameters that we choose in our simulations and experiments in Section~\ref{sec:hp-app} of the Supplementary Material.

\textit{Quadrature node count.} The closed-form AUC posterior requires a one- and a two-dimensional integral against $\mu$ and $K^+$. We use $q = 64$ Gauss--Legendre nodes for the numerical analyses, which we find to be accurate for the low-dimensional numerical integrals relevant to this work.  

\textit{Posterior path simulation.} For $C_{\max}$, the grid $U_q$ used for path simulation needs to be dense enough that the discrete maximum approximates the continuous-time supremum. Over a 24-hour observation window, $q = 200$ uniformly spaced nodes give a spacing of approximately $7$ minutes. This provides a comparatively fine starting grid relative to the clinical sampling schedules recommended for characterizing peaks in rapidly absorbed oral products \citep{ich2024m13a}. The grid resolution can be increased as needed to capture faster dynamics or sharper peaks.
The number of paths $S$ controls Monte Carlo error in the empirical $C_{\max}$ posterior; we use $S = 2000$ by default.

\textit{Computational complexity.} For a single curve with $n$ observations and a prediction grid of size $q$, Step 2 costs $O(n^3)$ and the covariance factorization in Step 6 costs $O(q^3)$. For typical pharmacokinetic studies with $n$ on the order of $5$--$20$ and $q$ on the order of $200$--$500$, these matrices are modest in size. Generating $S$ posterior paths in Step 6 additionally costs $O(S q^2)$.

\section{Simulation study}
\label{sec:sim}

We evaluate whether the proposed method delivers well-calibrated patient-level uncertainty for $\AUC$ and $C_{\max}$ across a range of pharmacokinetic settings. The design contrasts nine scenarios where simulated datasets are generated and analyzed by both our proposed model and alternative methods.

\subsection{Simulation design}

\paragraph{Scenarios.} The nine scenarios that we consider are listed in Table~\ref{tab:scenarios}. We group them as compartmental curves with standard sampling (S1--S3), compartmental curves with sparse sampling (S4--S6), and curves with an added discrepancy that makes them deviate from compartmental structure (S7--S9). Scenario~1 generates data from the same one-compartment oral model used as the prior mean, so that the mechanistic mean is correctly specified. Scenario~2 generates data from a two-compartment model while the fitted prior mean remains a one-compartment model, testing the GP's ability to absorb structural misspecification. Scenario~3 generates data from a one-compartment intravenous bolus model and fits the corresponding intravenous mean. Scenarios~4--6 mirror Scenarios~1--3 but with a sparse sampling design, to assess robustness when only a few post-dose measurements per patient are available. Scenarios~7--9 deviate from the exact compartmental model specification. The true curve includes an additive non-compartmental perturbation, so it does not belong to the compartmental family. Scenario~7 pairs the perturbation with a correctly specified one-compartment oral mean, Scenario~8 adds structural misspecification by generating from a two-compartment model, and Scenario~9 repeats Scenario~7 under the sparse design.

\begin{table}[!htbp]
\begingroup
\spacingset{1}
\caption{\textbf{Simulation scenarios to evaluate the compartment-driven GP model.} Sample size per subject $n_i$ is the number of post-dose concentration measurements. The standard design grid we use is $t \in \{0.5, 1, 1.5, 2, 3, 4, 6, 8, 12, 24\}$\,hr post-dose; the sparse design grid we use is $t \in \{1, 2, 4, 8, 16\}$\,hr. Here IV denotes intravenous administration.\label{tab:scenarios}}
\begin{center}
\begin{small}
\setlength{\tabcolsep}{4.5pt}
\begin{tabular}{lllll}
\toprule
Scenario & Generating model & Fitted mean & Sampling design & $n_i$ \\
\midrule
S1 & 1-compartment oral & 1-compartment oral & Standard grid & 10 \\
S2 & 2-compartment & 1-compartment oral & Standard grid & 10 \\
S3 & 1-compartment IV bolus & 1-compartment IV bolus & Standard grid & 10 \\
S4 & 1-compartment oral & 1-compartment oral & Sparse & 5 \\
S5 & 2-compartment & 1-compartment oral & Sparse & 5 \\
S6 & 1-compartment IV bolus & 1-compartment IV bolus & Sparse & 5 \\
\midrule
S7 & 1-compartment oral $+$ discrepancy & 1-compartment oral & Standard grid & 10 \\
S8 & 2-compartment $+$ discrepancy & 1-compartment oral & Standard grid & 10 \\
S9 & 1-compartment oral $+$ discrepancy & 1-compartment oral & Sparse & 5 \\
\bottomrule
\end{tabular}
\end{small}
\end{center}
\endgroup
\end{table}

\paragraph{Ground truth.} The population pharmacokinetic parameters for the simulations are fixed at $\kappa_a^{\text{pop}} = 1.5\,\text{hr}^{-1}$, $\kappa_e^{\text{pop}} = 0.1\,\text{hr}^{-1}$, and $c_0^{\text{pop}} = 10\,\text{mg/L}$. For each patient $i$, the patient-specific parameters are drawn as
\[
\kappa_{a,i} \sim \Normal(\kappa_a^{\text{pop}}, \sigma_{\kappa_a}^2), \qquad
\kappa_{e,i} \sim \Normal(\kappa_e^{\text{pop}}, \sigma_{\kappa_e}^2), \qquad
c_{0,i} \sim \Normal(c_0^{\text{pop}}, \sigma_{c_0}^2),
\]
with between-patient standard deviations $\sigma_{\kappa_a} = 0.3$, $\sigma_{\kappa_e} = 0.02$, and $\sigma_{c_0} = 2$. These numbers are chosen so that each parameter has a coefficient of variation of $20\%$.  
The two-compartment scenarios additionally specify inter-compartmental rate constants $\kappa_{12}^{\text{pop}} = 0.2\,\text{hr}^{-1}$ and $\kappa_{21}^{\text{pop}} = 0.1\,\text{hr}^{-1}$, with analogous Gaussian draws around these means. The observation model is $Y_{ij} = C_i(t_{ij}) + \epsilon_{ij}$ with $\epsilon_{ij} \sim \Normal(0, \sigma_\epsilon^2)$ and $\sigma_\epsilon = 0.1\,\text{mg/L}$. In Scenarios~7--9 there is an additional non-compartmental discrepancy,
$$
C_i(t) \;=\; \max\!\left\{\, m_i(t;\theta_i) + \delta_i(t),\; 0 \,\right\}, \qquad
\delta_i \overset{\text{iid}}{\sim} \mathcal{GP}\!\left(0,\; \tau^2 \rho_{5/2}(|s-t| / \ell_\delta)\right),
$$
drawn independently for each patient with $\tau = 0.6$\,mg/L and $\ell_\delta = 4$\,hr. 

\paragraph{Competing methods to the compartment-driven GP model.}
We list the competing methods compared below. 
\begin{enumerate}
\item Trapezoidal noncompartmental analysis (NCA+B) with a nonparametric bootstrap over time points, giving $\widehat{\AUC} = \sum_j \tfrac{1}{2}(t_{i,j+1}-t_{ij})(Y_{i,j+1}+Y_{ij})$ and $\widehat{C}_{\max} = \max_j Y_{ij}$ with bootstrap intervals.
\item Parametric one-compartment fit (P-1cpt).
\item Parametric two-compartment fit (P-2cpt).
\item Pure GP with a zero mean function, carrying no pharmacokinetic structure.
\item Compartment-driven GP (CD-GP; ours).
\end{enumerate}

\paragraph{Metrics for comparison.} Throughout the simulation the exposure window is $[a,b] = [0, 24]$\,hr, so $\AUC$ below means $\AUC_{0:24}$. We generate $B = 50$ simulated datasets of $N = 10$ patients each per scenario. For each simulated dataset $r$ and each patient $i$, we compute the point estimate $\widehat{\AUC}^{(r)}_i$ and a $95\%$ interval $\CI^{(r)}_i$ for $\AUC$,  
and analogously for $C_{\max}$. Writing $\AUC^{(r)}_{i,\text{true}}$ for the corresponding true value, and averaging over all $BN$ patient--dataset pairs, we report bias, root mean squared error (RMSE), coverage, and interval width:
\[
\begin{aligned}
\text{Bias}
&=\frac{1}{BN}\sum_r\sum_i
(\widehat{\AUC}^{(r)}_i-\AUC^{(r)}_{i,\text{true}}),\\
\text{RMSE}
&=\left[\frac{1}{BN}\sum_r\sum_i
(\widehat{\AUC}^{(r)}_i-\AUC^{(r)}_{i,\text{true}})^2\right]^{1/2}.
\end{aligned}
\]
$$
\text{Coverage} = \frac{1}{BN} \sum_{r} \sum_{i} \mathbf{1}\{\AUC^{(r)}_{i,\text{true}} \in \CI^{(r)}_i\},
\quad
\text{Width} = \frac{1}{BN} \sum_{r} \sum_{i} |\CI^{(r)}_i|,
$$
Here $|\CI|$ denotes interval length; the GP intervals are posterior credible intervals.

\subsection{Results}
We first assess uncertainty in total exposure across the three scenario groups. Across Scenarios~1--9, \AUC{} coverage for our method is approximately at or above the nominal $95\%$ (Table~\ref{tab:sim-results}, S1--S9, CD-GP rows). The intervals for our method are generally substantially narrower than those of the competitors that also attain high coverage. In the compartmental-curve scenarios S1--S6, coverage is conservative ($0.976$--$1.000$). This conservatism can be interpreted as a cost of robustness. The GP component accommodates departures from compartmental dynamics, but the additional flexibility can yield more conservative coverage when the compartmental structure already describes the data well. Related phenomena have been documented in the literature \citep{knapik2011bayesian,yang2017frequentist}.
For curves with an added discrepancy (Table~\ref{tab:sim-results}, S7--S9), the parametric compartmental fits exhibit undercoverage, whereas CD-GP maintains coverage much closer to the nominal level. 
The Pure GP competitor, which carries no pharmacokinetic structure, has conservative coverage and substantially wider \AUC{} credible intervals than CD-GP (Table~\ref{tab:sim-results}, S1--S9). These results support using the mechanistic mean to constrain exposure uncertainty while allowing departures from compartmental dynamics.

We next assess uncertainty in peak exposure in the oral-administration scenarios (Table~\ref{tab:sim-results}, S1--S2, S4--S5, and S7--S9). The noncompartmental and parametric compartmental methods generally undercover $C_{\max}$, whereas the Pure GP comparator yields conservative coverage and wider credible intervals than CD-GP. CD-GP generally achieves coverage closer to nominal, suggesting that combining mechanistic structure with a flexible discrepancy improves the balance between coverage and interval width in the settings examined.

\begingroup
\renewcommand{\baselinestretch}{1}
\small
\setlength{\tabcolsep}{4pt}
\renewcommand{\arraystretch}{1}
\setlength{\LTleft}{\fill}
\setlength{\LTright}{\fill}
\setlength{\LTcapwidth}{\textwidth}
\begin{longtable}{@{}llrrrrrrrr@{}}
\caption[Simulation results for $\AUC$ and $C_{\max}$ across Scenarios~S1--S9.]{\begingroup\renewcommand{\baselinestretch}{1.9}\selectfont
Simulation results for $\AUC$ and $C_{\max}$ across Scenarios~S1--S9.
Columns report bias, root mean squared error (RMSE), empirical coverage of nominal $95\%$ intervals, and mean interval width across $B\times N=500$ patient--dataset pairs; GP intervals are posterior credible intervals. Bias, RMSE, and width are in mg$\cdot$hr/L for $\AUC$ and mg/L for $C_{\max}$.
Nominal coverage is $0.95$; the Monte Carlo standard error of coverage near nominal is approximately $0.010$.
Methods: NCA+B, NCA with bootstrap; P-1cpt/P-2cpt, parametric one-/two-compartment fits; Pure GP, zero-mean GP; CD-GP, compartment-driven GP. Em dashes indicate $C_{\max}$ results omitted for the intravenous-bolus scenarios S3 and S6, whose generating curves decay monotonically from administration and attain their maximum at the initial concentration $c_0$.\endgraf\endgroup}
\label{tab:sim-results}\\
\toprule
& & \multicolumn{4}{c}{$\AUC$} & \multicolumn{4}{c}{$C_{\max}$} \\
\cmidrule(lr){3-6}\cmidrule(lr){7-10}
Scenario & Method & Bias & RMSE & Coverage & Width & Bias & RMSE & Coverage & Width \\
\midrule
\endfirsthead
\multicolumn{10}{c}{\tablename\ \thetable{} (continued)}\\
\toprule
& & \multicolumn{4}{c}{$\AUC$} & \multicolumn{4}{c}{$C_{\max}$} \\
\cmidrule(lr){3-6}\cmidrule(lr){7-10}
Scenario & Method & Bias & RMSE & Coverage & Width & Bias & RMSE & Coverage & Width \\
\midrule
\endhead
\midrule
\multicolumn{10}{r@{}}{\textit{Continued on next page}}\\
\endfoot
\bottomrule
\endlastfoot
\multicolumn{10}{@{}p{\textwidth}@{}}{A. Standard sampling: S1, one-compartment oral; S2, two-compartment; S3, intravenous bolus.}\\*
S1 & NCA+B & \phantom{$-$}1.12 & 1.59 & 1.000 & \phantom{0}59.9 & $-$0.00 & 0.09 & 0.514 & 0.61 \\*
 & P-1cpt & $-$0.06 & 1.03 & 0.768 & \phantom{0}\phantom{0}2.8 & $-$0.00 & 0.05 & 0.798 & 0.13 \\*
 & P-2cpt & \phantom{$-$}0.02 & 1.09 & 0.712 & \phantom{0}\phantom{0}2.6 & \phantom{$-$}0.00 & 0.06 & 0.692 & 0.13 \\*
 & Pure GP & $-$5.50 & 6.13 & 1.000 & 137.1 & \phantom{$-$}0.52 & 0.56 & 0.996 & 4.82 \\*
 & CD-GP & $-$0.35 & 3.09 & \textbf{0.998} & \phantom{0}17.7 & \phantom{$-$}0.01 & 0.12 & \textbf{0.936} & 0.40 \\
\midrule
S2 & NCA+B & $-$0.54 & 1.25 & 1.000 & \phantom{0}39.6 & $-$0.02 & 0.10 & 0.436 & 0.90 \\*
 & P-1cpt & $-$8.60 & 8.97 & 0.150 & \phantom{0}11.0 & $-$0.22 & 0.24 & 0.868 & 0.66 \\*
 & P-2cpt & $-$0.11 & 1.09 & 0.716 & \phantom{0}\phantom{0}2.6 & $-$0.01 & 0.06 & 0.694 & 0.14 \\*
 & Pure GP & $-$3.78 & 4.08 & 1.000 & 105.5 & \phantom{$-$}0.19 & 0.25 & 1.000 & 3.32 \\*
 & CD-GP & \phantom{$-$}1.61 & 2.26 & 1.000 & \phantom{0}23.3 & \phantom{$-$}0.04 & 0.12 & \textbf{0.922} & 0.40 \\
\midrule
S3 & NCA+B & $-$2.19 & 2.48 & 0.990 & \phantom{0}57.8 & --- & --- & --- & --- \\*
 & P-1cpt & \phantom{$-$}0.02 & 0.96 & 0.814 & \phantom{0}\phantom{0}2.8 & --- & --- & --- & --- \\*
 & P-2cpt & $-$0.14 & 2.69 & 0.708 & \phantom{0}\phantom{0}5.3 & --- & --- & --- & --- \\*
 & Pure GP & $-$5.94 & 6.51 & 1.000 & 134.2 & --- & --- & --- & --- \\*
 & CD-GP & $-$0.53 & 3.13 & 1.000 & \phantom{0}17.3 & --- & --- & --- & --- \\
\midrule
\multicolumn{10}{@{}p{\textwidth}@{}}{B. Sparse sampling: S4, one-compartment oral; S5, two-compartment; S6, intravenous bolus.}\\*
S4 & NCA+B & $-$16.15 & 17.62 & 0.044 & \phantom{0}59.9 & $-$0.04 & 0.13 & 0.372 & 1.33 \\*
 & P-1cpt & $-$0.01 & 1.28 & 0.648 & \phantom{0}\phantom{0}2.6 & $-$0.01 & 0.07 & 0.560 & 0.14 \\*
 & P-2cpt & \phantom{$-$}0.22 & 1.38 & 0.394 & \phantom{0}\phantom{0}1.6 & \phantom{$-$}0.01 & 0.09 & 0.368 & 0.09 \\*
 & Pure GP & $-$4.08 & 5.22 & 1.000 & 118.7 & \phantom{$-$}0.85 & 0.93 & 0.994 & 5.62 \\*
 & CD-GP & $-$0.61 & 3.91 & \textbf{0.976} & \phantom{0}16.6 & $-$0.01 & 0.13 & \textbf{0.976} & 0.55 \\
\midrule
S5 & NCA+B & $-$13.02 & 13.50 & 0.752 & \phantom{0}47.3 & $-$0.14 & 0.19 & 0.144 & 2.35 \\*
 & P-1cpt & $-$5.72 & 6.08 & 0.720 & \phantom{0}14.4 & $-$0.02 & 0.40 & 0.936 & 1.44 \\*
 & P-2cpt & \phantom{$-$}0.01 & 1.61 & 0.122 & \phantom{0}\phantom{0}0.6 & \phantom{$-$}0.01 & 0.08 & 0.146 & 0.03 \\*
 & Pure GP & $-$3.33 & 3.74 & 1.000 & \phantom{0}89.9 & \phantom{$-$}0.35 & 0.39 & 1.000 & 3.58 \\*
 & CD-GP & \phantom{$-$}0.85 & 2.00 & 1.000 & \phantom{0}20.9 & $-$0.01 & 0.12 & \textbf{0.992} & 0.59 \\
\midrule
S6 & NCA+B & $-$20.07 & 21.10 & 0.014 & \phantom{0}60.1 & --- & --- & --- & --- \\*
 & P-1cpt & \phantom{$-$}0.05 & 1.18 & 0.734 & \phantom{0}\phantom{0}3.1 & --- & --- & --- & --- \\*
 & P-2cpt & \phantom{$-$}1.22 & 10.28 & 0.456 & \phantom{0}12.4 & --- & --- & --- & --- \\*
 & Pure GP & $-$6.08 & 6.75 & 1.000 & 127.9 & --- & --- & --- & --- \\*
 & CD-GP & $-$0.53 & 3.58 & \textbf{0.980} & \phantom{0}16.1 & --- & --- & --- & --- \\
\midrule
\multicolumn{10}{@{}p{\textwidth}@{}}{C. Added non-compartmental discrepancy: S7/S9, one-compartment oral; S8, two-compartment; S9, sparse.}\\*
S7 & NCA+B & \phantom{$-$}1.43 & 4.05 & 0.964 & \phantom{0}62.7 & $-$0.00 & 0.10 & 0.482 & 0.64 \\*
 & P-1cpt & \phantom{$-$}0.16 & 4.26 & 0.600 & \phantom{0}\phantom{0}8.3 & $-$0.01 & 0.15 & 0.782 & 0.39 \\*
 & P-2cpt & \phantom{$-$}0.58 & 4.03 & 0.532 & \phantom{0}\phantom{0}6.3 & \phantom{$-$}0.01 & 0.11 & 0.804 & 0.32 \\*
 & Pure GP & $-$5.43 & 7.04 & 1.000 & 140.8 & \phantom{$-$}0.55 & 0.59 & 0.996 & 5.00 \\*
 & CD-GP & $-$0.32 & 4.76 & \textbf{0.952} & \phantom{0}18.7 & $-$0.00 & 0.12 & \textbf{0.920} & 0.41 \\
\midrule
S8 & NCA+B & $-$0.63 & 4.10 & 0.968 & \phantom{0}41.3 & $-$0.02 & 0.09 & 0.380 & 0.89 \\*
 & P-1cpt & $-$8.11 & 10.04 & 0.354 & \phantom{0}12.3 & $-$0.22 & 0.27 & 0.816 & 0.74 \\*
 & P-2cpt & $-$0.36 & 4.29 & 0.478 & \phantom{0}\phantom{0}6.0 & \phantom{$-$}0.01 & 0.09 & 0.838 & 0.32 \\*
 & Pure GP & $-$3.96 & 5.61 & 1.000 & 105.5 & \phantom{$-$}0.20 & 0.27 & 1.000 & 3.38 \\*
 & CD-GP & \phantom{$-$}1.41 & 4.23 & \textbf{0.998} & \phantom{0}24.1 & \phantom{$-$}0.04 & 0.13 & \textbf{0.922} & 0.42 \\
\midrule
S9 & NCA+B & $-$16.13 & 17.57 & 0.096 & \phantom{0}60.9 & $-$0.06 & 0.13 & 0.294 & 1.36 \\*
 & P-1cpt & \phantom{$-$}0.01 & 3.65 & 0.670 & \phantom{0}\phantom{0}9.1 & \phantom{$-$}0.00 & 0.18 & 0.740 & 0.49 \\*
 & P-2cpt & \phantom{$-$}0.97 & 4.41 & 0.370 & \phantom{0}\phantom{0}4.8 & \phantom{$-$}0.04 & 0.14 & 0.482 & 0.80 \\*
 & Pure GP & $-$4.02 & 5.62 & 1.000 & 119.8 & \phantom{$-$}0.85 & 0.92 & 0.988 & 5.66 \\*
 & CD-GP & $-$0.58 & 4.32 & \textbf{0.968} & \phantom{0}17.1 & $-$0.01 & 0.14 & \textbf{0.970} & 0.56 \\
\end{longtable}
\endgroup

\section{Case studies}
\label{sec:case}

We now illustrate the compartment-driven GP framework on several real pharmacokinetic datasets that span the two data types described in Section~\ref{sec:notation}. The first case concerns individualized patient-level time-series via the warfarin dataset, distributed in the \texttt{nlmixr2data} R package. This dataset is commonly used as a benchmark in pharmacokinetic modeling. A second individualized analysis, of the theophylline dataset from the same package, is reported in Section~\ref{sec:theo} of the Supplementary Material. 
The remaining two cases are aggregated concentration--time curves digitized from a clinical-pharmacology study of the chemotherapy agents capecitabine and 5-fluorouracil (5-FU) in patients with normal liver function and in patients with hepatic dysfunction due to liver metastases \citep{twelves1999hepatic}. For each chemotherapy drug, the framework produces posterior distributions over the cohort-mean \AUC{} for the two patient populations. This allows for a probabilistic assessment of the extent to which the hepatic-dysfunction cohort received higher total exposure to 5-FU and capecitabine than the cohort with normal liver function.  

Throughout this section, the prior mean function is a one-compartment oral absorption model as in Equation~(\ref{eq:onecpt}) with population-typical rate parameters taken from the pharmacokinetic literature. All case studies use the Mat\'ern-5/2 kernel with the length scale held fixed and the remaining hyperparameter priors (Supplementary Material Section~\ref{sec:hp-app}) scaled to the concentration range of each drug. Posterior path simulation uses $S = 2000$ paths per case on a $200$-point grid spanning the inferential window.  

\subsection{Warfarin}
\label{sec:warfarin}

The warfarin dataset includes $32$ adult patients (ages 21--63, body weights 40--102\,kg) who each received a single oral loading dose of warfarin between 60 and 153\,mg (weight-scaled), followed by 6--17 plasma-concentration measurements over the subsequent 120 hours (Supplementary Figure~\ref{fig:warfarin-patients}). Warfarin is an oral anticoagulant with a narrow therapeutic window. Population-typical adult values $\kappa_a = 0.8\,\text{hr}^{-1}$ and $\kappa_e = 0.012\,\text{hr}^{-1}$ (half-life $\approx 58\,$hr) were used to inform the mechanistic prior mean. The administered dose enters on the concentration scale through $c_{0,i} = D_i / (v_{\text{kg}} W_i)$, where $D_i$ and $W_i$ are patient $i$'s dose and body weight and $v_{\text{kg}} = 0.14\,\text{L/kg}$ is the apparent volume of distribution reported for warfarin.

Figure~\ref{fig:warfarin-intervals} shows patient-level posterior 95\% credible intervals for AUC over $[t_{\text{dose}},t_{\text{dose}}+168\,\mathrm{hr}]$ and for $C_{\max}$, with posterior median exposures spanning a broad range across patients. In particular, $\AUC$ medians range from 483 to 961\,mg$\cdot$hr/L (median 796), and $C_{\max}$ medians range from 8.8 to 17.3\,mg/L (median 10.9). Two patients (IDs 8 and 12) have markedly higher $C_{\max}$ intervals than the others. These higher peak exposures may be relevant when clinicians assess bleeding risk and consider dose reduction. The posterior credible intervals provide uncertainty quantification that could help guide such clinical decisions. 

\begin{figure}[!htbp]
\begingroup
\spacingset{1}
\centering
\includegraphics[width=\linewidth]{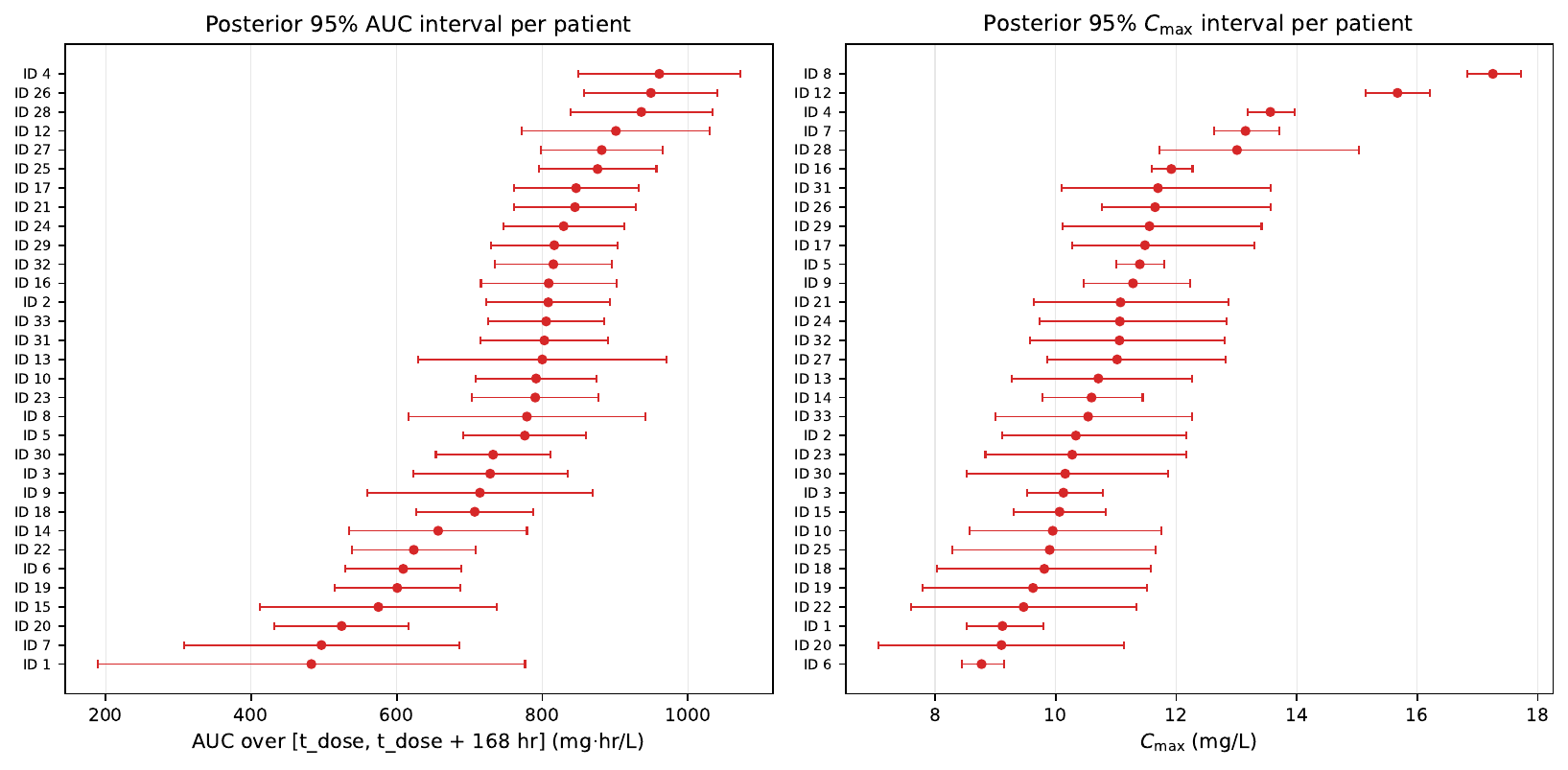}

\caption{Warfarin case study: posterior 95\% credible intervals per patient for $\AUC$ over $[t_{\text{dose}}, t_{\text{dose}} + 168\,\text{hr}]$ (left) and for $C_{\max}$ (right). Patients are sorted by posterior median.\label{fig:warfarin-intervals}}
\endgroup
\end{figure}

\subsection{Chemotherapy: 5-fluorouracil and capecitabine}
\label{sec:chemo}

We now consider aggregated concentration--time data for two related chemotherapy agents, digitized from the clinical-pharmacology study of \citet{twelves1999hepatic}. That study reports plasma concentration--time profiles of capecitabine and its principal active metabolite 5-fluorouracil (5-FU) in 27 patients with cancer who were evaluable for pharmacokinetic analysis: 14 patients with normal liver function and 13 patients with mild-to-moderate hepatic dysfunction due to liver metastases. 5-FU is one of the most widely used cytotoxic agents in treating colorectal cancer and several other tumor types. It has a narrow therapeutic window and substantial elimination through hepatic metabolism. Capecitabine is an oral prodrug that is metabolized to 5-FU. It is given as a chemotherapeutic agent in colorectal, breast, and gastric cancers. Because capecitabine itself is eliminated by hepatic metabolism, liver function is expected to be a determinant of plasma exposure.  Our framework converts the cohort-mean concentration--time curves and standard errors at observed time points into posterior distributions for each group's AUC. These distributions allow us to probabilistically assess whether exposure is higher in the hepatic dysfunction group.

Posterior fits to the cohort-mean concentration--time curves are shown in Section~\ref{sec:support-chemo} of the Supplementary Material. Comparing the posterior distributions suggests higher exposure in the hepatic dysfunction cohort for both 5-FU and capecitabine (Figure~\ref{fig:chemo-auc}). We estimate the posterior probability that $\mathrm{AUC}^{\mathrm{dysfunction}}>\mathrm{AUC}^{\mathrm{normal}}$ by repeatedly drawing AUC values from each cohort's posterior distribution and calculating the proportion of comparisons in which the dysfunction-cohort value is larger. For 5-FU, the posterior probability is approximately 0.77, with a posterior mean difference of approximately $+0.10\,\mu\mathrm{g}\cdot\mathrm{hr}/\mathrm{mL}$. For capecitabine, the posterior probability is approximately 0.93, with a posterior mean difference of approximately $+3.9\,\mu\mathrm{g}\cdot\mathrm{hr}/\mathrm{mL}$. The posterior comparisons favor a positive difference for both drugs. These findings highlight the potential value of accounting for liver function when assessing exposure and considering individualized dosing for 5-FU and capecitabine. 

\begin{figure}[!htbp]
\begingroup
\spacingset{1}
\centering
\includegraphics[width=\linewidth]{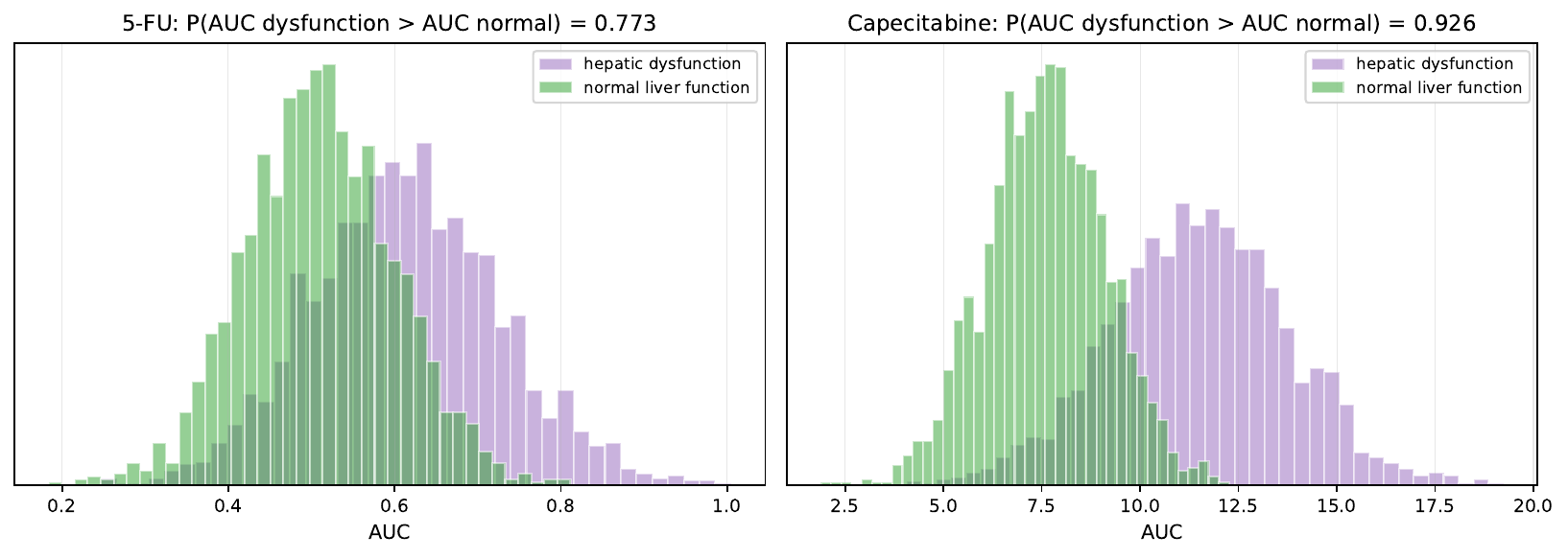}

\caption{Chemotherapy case study: posterior densities of cohort-mean \AUC{} for the hepatic dysfunction cohort (purple) and the normal liver function cohort (green), shown for 5-FU (left) and capecitabine (right). The title of each panel reports the estimated posterior probability that cohort-mean AUC is higher in the hepatic dysfunction group than in the normal liver function group. AUC is measured in $\mu\mathrm{g}\cdot\mathrm{hr}/\mathrm{mL}$.\label{fig:chemo-auc} 
}
\endgroup
\end{figure}

\section{Discussion}
\label{sec:disc}

We have presented a general probabilistic framework for uncertainty quantification of functionals of pharmacokinetic concentration--time curves, combining a compartment model-driven prior with a Gaussian process component. Although \AUC{} and $C_{\max}$ are the two functionals of primary clinical interest, the posterior over the latent concentration--time curve can be propagated to many other pharmacokinetic summary statistics. Examples include clearance, elimination half-life and time above a therapeutic or toxicity threshold. While some functionals may not admit closed-form posterior distributions, their uncertainty can be evaluated by applying the relevant functional to posterior draws of the latent curve.

Our approach reveals a deliberate trade-off between precision and robustness. If a compartmental model fully describes the data-generating process, a purely parametric compartmental analysis can produce narrower intervals by exploiting that structure. When the compartmental dynamics are only approximate, however, such intervals may be overconfident because they do not account for misspecification. Our model allows for added flexibility that alleviates undercoverage under misspecification, but may yield more conservative coverage when the compartmental model is correct. 

Several future directions could broaden the framework further. A hierarchical formulation could jointly analyze multiple patients, allowing information to be shared across individuals while retaining patient-specific posterior exposure summaries. Nonstationary covariance structures may also be useful because the rapid absorption phase and slower elimination phase may have substantially different temporal scales. The posterior uncertainty quantified by the framework could be used prospectively rather than only for retrospective analysis. For example, it can provide a basis for selecting sampling times that are most informative about \AUC{} and $C_{\max}$, for adapting measurement schedules to individual patients, and potentially for evaluating dose-adjustment decisions under explicit efficacy--toxicity trade-offs.

\section*{Acknowledgements}
ET acknowledges funding support from the Warren Alpert Foundation and the Croucher Foundation. BEE was funded in part by grants from the Parker Institute for Cancer Immunology (PICI), the Chan-Zuckerberg Institute (CZI), NIH NHGRI R01 HG012967, and NIH NHGRI R01 HG013736. BEE is a CIFAR Fellow in the Multiscale Human Program. BEE is on the Scientific Advisory Board for ArrePath Inc and Freenome; she consults for Neumora.

\section*{AI Use Disclosure}

AI tools (GPT 6 Astra and Claude Opus 5.5) assisted with coding, refining text, and the derivation of mathematical proofs. The authors are responsible for the content of this manuscript.

\clearpage

\begin{center}
{\large\bf SUPPLEMENTARY MATERIAL}
\end{center}

\bigskip

\noindent This Supplementary Material accompanies the main paper. Section~\ref{sec:hp-app} records application-specific hyperparameters and numerical settings. Section~\ref{sec:auc-numerics} introduces quadrature and gives the numerical AUC formulas and convergence results. Section~\ref{sec:proofs} provides proofs of the main-text theorems and supporting results for maximum functionals. Section~\ref{sec:support-cases} collects supporting simulation and application figures and the additional theophylline analysis.

\setcounter{section}{0}
\renewcommand{\thesection}{S\arabic{section}}
\renewcommand{\thesubsection}{\thesection.\arabic{subsection}}
\setcounter{theorem}{0}
\renewcommand{\thetheorem}{\thesection.\arabic{theorem}}
\setcounter{table}{0}
\renewcommand{\thetable}{S\arabic{table}}
\setcounter{figure}{0}
\renewcommand{\thefigure}{S\arabic{figure}}
\setcounter{equation}{0}
\renewcommand{\theequation}{S\arabic{equation}}

\section{Hyperparameter prior specification across applications}
\label{sec:hp-app}

We report here the hyperparameter priors used across the simulation study and the four case studies. In all applications, we place a log-normal prior on the kernel standard deviation $\sigma_f$ and estimate it by MAP penalized marginal likelihood (Section~\ref{sec:hyper}). For individual-level data, we also estimate the observation-noise standard deviation $\sigma_\epsilon$ by MAP, assigning it an independent log-normal prior. For aggregated data, the observation-noise standard deviation at each observed time point is fixed at the reported standard error. All log-normal priors have a common log-standard-deviation $\varsigma = 0.2$. Table~\ref{tab:hp-priors} reports the prior medians used in each application.

\begin{table}[h]
\caption{Prior medians $s_{\sigma_f}$ and $s_{\sigma_\epsilon}$ for the log-normal hyperparameter priors (common log-SD $\varsigma = 0.2$), together with the fixed length scale $\ell$, used in the simulation study and the case studies. ``Reported SE'' denotes the reported standard error at each observed time point in aggregated data. Concentrations are reported in the units used by the source of each dataset; note that mg/L and $\mu$g/mL are numerically identical. \label{tab:hp-priors}}
\begin{center}
\begin{small}
\begin{tabular}{lllccc}
\toprule
Application & Units & Kernel & $s_{\sigma_f}$ & $s_{\sigma_\epsilon}$ & $\ell$ (fixed) \\
\midrule
Simulation S1--S9        & mg/L    & Mat\'ern-5/2 & 0.6 & 0.2          & 4\,hr \\
Warfarin                 & mg/L    & Mat\'ern-5/2 & 1.5 & 0.2          & 24\,hr \\
Theophylline             & mg/L    & Mat\'ern-5/2 & 1.5 & 0.2          & 0.5\,hr \\
5-Fluorouracil (aggregated)   & $\mu$g/mL  & Mat\'ern-5/2 & 0.3 & reported \SEr  & 1.5\,hr \\
Capecitabine (aggregated)     & $\mu$g/mL  & Mat\'ern-5/2 & 4.0 & reported \SEr  & 1.5\,hr \\
\bottomrule
\end{tabular}
\end{small}
\end{center}
\end{table}

The prior medians in Table~\ref{tab:hp-priors} were chosen by inspection of the concentration scale and sampling cadence of each dataset. The prior median for the kernel standard deviation $\sigma_f$ is set to be of the same order as the empirical standard deviation of the residuals $Y_{ij} - m_i(t_{ij};\theta^{\text{pop}})$ between the observed concentrations and the population-typical mechanistic prior mean. For the warfarin and theophylline datasets, both of which are on the order of approximately $10$\,mg/L, this prescription gives $s_{\sigma_f} = 1.5$\,mg/L. For the chemotherapy aggregated cohort-mean curves, the concentration scales differ between 5-fluorouracil (peak cohort-mean of approximately $0.2$\,$\mu$g/mL, with maximum standard error of approximately $0.07$\,$\mu$g/mL) and capecitabine (peak cohort-mean of approximately $5.6$\,$\mu$g/mL, with maximum standard error of approximately $1.8$\,$\mu$g/mL), and the prior medians are scaled accordingly to $0.3$ and $4.0$\,$\mu$g/mL.

The prior median for the observation-noise standard deviation $\sigma_\epsilon$ is set at a level consistent with reported assay precision in the relevant plasma-pharmacokinetic literature for the warfarin and theophylline datasets, and at a comparable level for the simulation study, giving $\sigma_\epsilon$ prior median $0.2$\,mg/L throughout the individualized cases. For the aggregated chemotherapy case studies, the observation-noise standard deviations are fixed at the reported per-timepoint standard errors, so no prior on $\sigma_\epsilon$ is needed.

The length scale $\ell$ is not estimated. Across applications considered in this study it is held fixed at the value in Table~\ref{tab:hp-priors}. For the chemotherapy cohort-mean curves we choose $\ell = 1.5$\,hr based on the observation cadence, which is on the order of $1$ to $2$ hours around the absorption phase. For theophylline a short $\ell = 0.5$\,hr is chosen. For warfarin, we use $\ell = 24$\,hr, which is well inside the $58$\,hr elimination half-life of the drug. 

For the simulation study we use the Mat\'ern-5/2 kernel with $\sigma_f$ prior median $s_{\sigma_f} = 0.6$\,mg/L and the length scale fixed at $\ell = 4$\,hr. In Scenarios~7--9 the discrepancy of the data-generating process uses the same scale and length scale, $\tau = 0.6$\,mg/L and $\ell_\delta = 4$\,hr.

\paragraph{Exposure windows.}
The chemotherapy analyses compute AUC over $[0,7]$\,hr for 5-FU and $[0,6]$\,hr for capecitabine.

\section{Quadrature for AUC inference}
\label{sec:auc-numerics}

Numerical quadrature is an approach that approximates an integral by a weighted sum of function evaluations. Given locations (known as nodes) $u_1,\ldots,u_q$ and weights $\omega_1,\ldots,\omega_q$, a quadrature rule approximates $\int_a^b f(t)\,dt$ by $\sum_{v=1}^q\omega_v f(u_v)$. The nodes specify where the function is evaluated, and the weights determine each evaluation's contribution to the integral. Gaussian quadrature selects nodes and weights using polynomials to obtain accurate integration of smooth functions.

Writing $U_q=(u_1,\ldots,u_q)^\top$ and $\omega_q=(\omega_1,\ldots,\omega_q)^\top$, and recalling that $C$ denotes the concentration--time curve, the quadrature approximation to the AUC is
\begin{equation}
A_q=\omega_q^\top C(U_q).
\label{eq:quad-auc}
\end{equation}
Applying the rule to the posterior mean and covariance in Theorem~\ref{thm:auc-gaussian} gives
\[
\mu_{A,q}=\omega_q^\top\mu(U_q),\qquad
\sigma_{A,q}^2=\omega_q^\top K^+(U_q,U_q)\omega_q.
\]
These calculations approximate the integrals defining the Gaussian AUC posterior without requiring posterior path simulation. Gaussianity follows from the linearity of the functional, while the quadrature approach provides a way to numerically evaluate the resulting integral. The following results give the distribution of $A_q$ and its consistency for estimating $\AUC_{a:b}$.

\subsection{Quadrature \AUC{} distribution}

\begin{theorem}[Quadrature \AUC{} distribution]
\label{thm:S-quad-auc}
Let $U_q = (u_1, \ldots, u_q)^\top$ be quadrature nodes and $\omega_q = (\omega_1, \ldots, \omega_q)^\top$ be quadrature weights. Define $A_q = \omega_q^\top C(U_q) = \sum_{v=1}^q \omega_v C(u_v)$. Then
\[
A_q \mid Y = y \sim \Normal\!\left(\omega_q^\top \mu(U_q),\, \omega_q^\top K^+(U_q, U_q) \omega_q\right).
\]
\end{theorem}

\begin{proof}
By Theorem~\ref{thm:S-gp-posterior}, $C(U_q) \mid Y = y \sim \Normal_q(\mu(U_q), K^+(U_q, U_q))$. Then $A_q = \omega_q^\top C(U_q)$ is a linear combination of a multivariate normal vector, hence normally distributed with
\[
\E(A_q \mid Y = y) = \omega_q^\top \mu(U_q),
\qquad
\Var(A_q \mid Y = y) = \omega_q^\top K^+(U_q, U_q) \omega_q.
\]
\end{proof}

\subsection{Quadrature consistency}

\begin{theorem}[Quadrature consistency]
\label{thm:S-quad-consistency}
Assume $C \mid Y = y$ has almost surely continuous sample paths on $[a, b]$. Let $Q_q(f) = \sum_{v=1}^q \omega_{vq} f(u_{vq})$ be a sequence of quadrature rules satisfying:
\begin{itemize}
\item[(i)] $Q_q(f) \to \int_a^b f(t)\, dt$ for every continuous function $f$ on $[a, b]$;
\item[(ii)] $\Lambda := \sup_q \sum_{v=1}^q |\omega_{vq}| < \infty$.
\end{itemize}
Define $A_q = Q_q(C)$ and $A = \int_a^b C(t)\, dt$. Then $A_q \to A$ almost surely. If additionally $\E[\sup_{t \in [a,b]} |C(t)| \mid Y = y] < \infty$, then $A_q \to A$ in $L^1$. If $\E[\sup_{t \in [a,b]} |C(t)|^2 \mid Y = y] < \infty$, then $A_q \to A$ in $L^2$.
\end{theorem}

\begin{proof}
For almost every sample path, $t \mapsto C(t)$ is continuous on $[a, b]$. Hence by (i), $Q_q(C) \to \int_a^b C(t)\, dt$ for that path, so $A_q \to A$ almost surely.

For $L^p$ convergence, bound the deviation pathwise:
\[
|A_q| = \Big|\sum_{v=1}^q \omega_{vq} C(u_{vq})\Big|
\le \Big(\sum_{v=1}^q |\omega_{vq}|\Big) \sup_{t \in [a,b]} |C(t)|
\le \Lambda \sup_t |C(t)|,
\]
and $|A| \le (b - a) \sup_t |C(t)|$. Thus
\[
|A_q - A| \le \{\Lambda + (b - a)\} \sup_t |C(t)|.
\]
If $\sup_t |C(t)|$ is integrable, the dominated convergence theorem applied to $|A_q - A|$ gives $L^1$ convergence. If $\sup_t |C(t)|^2$ is integrable, $|A_q - A|^2$ is dominated by an integrable random variable (with constant $\{\Lambda + (b - a)\}^2$), and dominated convergence yields $L^2$ convergence.
\end{proof}

\section{Proofs of theoretical results}
\label{sec:proofs}

We now state supplemental theoretical results as well as related proofs. Throughout this section we adopt the notation of Section~\ref{sec:model} of the main paper. Unless stated otherwise, $C \sim \GP(m, k)$ on $\Tcal = [a, b]$, observations satisfy $Y = C(T) + \epsilon$ with $\epsilon \sim \Normal(0, R)$ and $R$ positive definite, and we condition on $Y = y$. The posterior is $C \mid Y = y \sim \GP(\mu, K^+)$ with
\[
\begin{aligned}
\mu(t)
  &=m(t)+k(t,T)(K_{TT}+R)^{-1}\{y-m(T)\},\\
K^+(s,t)
  &=k(s,t)-k(s,T)(K_{TT}+R)^{-1}k(T,t),
\end{aligned}
\]
where $K_{TT} = k(T, T)$ and $k(t, T) = (k(t, t_1), \ldots, k(t, t_n))$.

\subsection{Posterior GP conjugacy}
For self-containment, we give a standard result below on the posterior distribution of a Gaussian process. See \citet{rasmussen2006gaussian} for an authoritative reference. 
\begin{theorem}[Posterior Gaussian process; corresponds to Theorem~\ref{thm:gp-posterior} of the main paper]
\label{thm:S-gp-posterior}
Assume $C \sim \GP(m, k)$ and $Y = C(T) + \epsilon$ with $\epsilon \sim \Normal(0, R)$, $R$ positive definite, and $\epsilon$ independent of $C$. Let $K_{TT} = k(T, T)$, $k(t, T) = \{k(t, t_1), \ldots, k(t, t_n)\}$, and $m(T) = (m(t_1), \ldots, m(t_n))^\top$. Then
\[
C \mid Y = y \sim \GP(\mu, K^+),
\]
where $\mu(t) = m(t) + k(t, T)(K_{TT} + R)^{-1}\{y - m(T)\}$ and $K^+(s, t) = k(s, t) - k(s, T)(K_{TT} + R)^{-1} k(T, t)$.
\end{theorem}

\begin{proof}
Let $U = (u_1, \ldots, u_q)$ be any finite set of prediction points. By the definition of a Gaussian process, $C(U) = \{C(u_1), \ldots, C(u_q)\}^\top$ and $C(T)$ are jointly Gaussian. Since $Y = C(T) + \epsilon$ with $\epsilon \sim \Normal(0, R)$ independent of $C$, the stacked vector
\[
\begin{pmatrix} C(U) \\ Y \end{pmatrix}
\]
is jointly Gaussian with mean $(m(U)^\top, m(T)^\top)^\top$ and covariance
\[
\begin{pmatrix}
K_{UU} & K_{UT} \\
K_{TU} & K_{TT} + R
\end{pmatrix}.
\]
The standard multivariate normal conditioning formula gives
\[
\begin{aligned}
C(U)\mid Y=y\sim\Normal\!\Big(
 &m(U)+K_{UT}(K_{TT}+R)^{-1}\{y-m(T)\},\\
 &K_{UU}-K_{UT}(K_{TT}+R)^{-1}K_{TU}\Big).
\end{aligned}
\]
This holds for every finite collection of prediction points $U$, so the conditional finite-dimensional distributions are consistent. By the Kolmogorov extension theorem, these finite-dimensional distributions determine a Gaussian process with mean function $\mu$ and covariance function $K^+$ as stated.
\end{proof}

\subsection{Gaussianity of \AUC}
We now give a result that certifies the Gaussianity of the AUC under our model. 
\begin{theorem}[Gaussianity of AUC; corresponds to Theorem~\ref{thm:auc-gaussian} of the main paper]
\label{thm:S-auc-gaussian}
Under the posterior in Equation~\ref{eq:gp-posterior}, assume that the mean and covariance functions are continuous and that $C$ has almost surely continuous sample paths on the compact interval $[a,b]$, with $a<b$. Then
\[
\AUC_{a:b}\mid Y=y\sim\Normal(\mu_A,\sigma_A^2),
\]
where
\[
\mu_A=\int_a^b\mu(t)\,dt,
\qquad
\sigma_A^2=\int_a^b\int_a^b K^+(s,t)\,ds\,dt.
\]
\end{theorem}

\begin{proof}
All expectations below are conditional on $Y=y$. Write $C(t)=\mu(t)+X(t)$, where $X\sim\GP(0,K^+)$. Continuous sample paths ensure that $W=\int_a^bX(t)\,dt$ exists almost surely, and $\AUC_{a:b}=\mu_A+W$.

Choose partitions of $[a,b]$ whose mesh tends to zero. Write $h_{j,n}$ for the length of the $j$th interval and choose a tag $\xi_{j,n}$ in that interval. Define the piecewise constant process $X_n(t)=X(\xi_{j,n})$ on the $j$th interval and
\[
W_n=\int_a^bX_n(t)\,dt
   =\sum_jh_{j,n}X(\xi_{j,n}).
\]
Each $W_n$ is a finite linear combination of jointly Gaussian evaluations. Consequently, $W_n\sim\Normal(0,v_n)$, where
\[
v_n=\sum_j\sum_\ell
h_{j,n}h_{\ell,n}K^+(\xi_{j,n},\xi_{\ell,n})
\longrightarrow\sigma_A^2
\]
by continuity of $K^+$.

Uniform continuity of $K^+$ on $[a,b]^2$ implies
\[
\sup_{t\in[a,b]}\E|X_n(t)-X(t)|^2\longrightarrow0.
\]
The Cauchy--Schwarz inequality and Tonelli's theorem therefore give
\[
\E|W_n-W|^2
\le(b-a)\int_a^b\E|X_n(t)-X(t)|^2\,dt
\longrightarrow0.
\]
Thus $W_n\to W$ in $L^2$ and hence in distribution. Since $v_n\to\sigma_A^2$, the Gaussian characteristic functions converge to $\exp(-u^2\sigma_A^2/2)$. Therefore $W\sim\Normal(0,\sigma_A^2)$. Adding $\mu_A$ proves the result.
\end{proof}

\subsection{Maximum functionals}
\label{sec:S-max-functionals}
We now give relevant theoretical results concerning $C_{\max}$. We first give an existence result. 
\begin{theorem}[Existence of $C_{\max}$]
\label{thm:S-cmax-exists}
If $C \mid Y = y$ has almost surely continuous sample paths on the compact interval $\Tcal = [a, b]$, then $C_{\max} = \sup_{t \in \Tcal} C(t)$ is finite almost surely and is attained almost surely; that is, $\argmax_{t \in \Tcal} C(t)$ is nonempty almost surely.
\end{theorem}

\begin{proof}
For almost every sample path, $t \mapsto C(t)$ is continuous on the compact set $[a, b]$. By the extreme value theorem, every continuous function on a compact set is bounded and attains its maximum. Hence for almost every sample path, $C_{\max} = \max_{t \in [a, b]} C(t) < \infty$ and $\argmax_{t \in [a, b]} C(t) \neq \varnothing$.
\end{proof}
 
We now give a result regarding our approach to estimating $C_{\max}$ via grids. 

\begin{theorem}[Finite-grid $C_{\max}$ distribution]
\label{thm:S-finitegrid-cmax}
Let $U_q = (u_1, \ldots, u_q)^\top$ and suppose $C(U_q) \mid Y = y \sim \Normal_q(\mu_q, K_q)$. Define $M_q = \max_{v=1, \ldots, q} C(u_v)$. Then
\[
\Prob(M_q \le x \mid Y = y) = \Phi_q(x \mathbf{1}_q;\, \mu_q, K_q),
\]
where $\Phi_q(\cdot;\, \mu_q, K_q)$ is the $q$-dimensional Gaussian CDF and $x \mathbf{1}_q = (x, \ldots, x)^\top$.
\end{theorem}

\begin{proof}
The event $\{M_q \le x\}$ is equivalent to $\bigcap_v \{C(u_v) \le x\}$, that is, $\{C(U_q) \le x \mathbf{1}_q\}$ componentwise. Since $C(U_q) \mid Y = y \sim \Normal_q(\mu_q, K_q)$, this probability is by definition $\Phi_q(x \mathbf{1}_q;\, \mu_q, K_q)$.
\end{proof}
We now give two convergence results. 
\begin{theorem}[Grid convergence of $C_{\max}$]
\label{thm:S-grid-conv}
Suppose $C \mid Y = y$ has almost surely continuous sample paths on $\Tcal = [a, b]$. Let $G_q \subset \Tcal$ be a sequence of grids satisfying
\[
\Delta(G_q) := \sup_{t \in \Tcal} \min_{u \in G_q} |t - u| \to 0,
\]
and define $M_q = \max_{u \in G_q} C(u)$ and $C_{\max} = \sup_{t \in \Tcal} C(t)$. Then $M_q \to C_{\max}$ almost surely. If additionally $\E[\sup_t |C(t)| \mid Y = y] < \infty$, then $M_q \to C_{\max}$ in $L^1$.
\end{theorem}

\begin{proof}
Fix a sample path $C(\cdot)$ that is continuous on $[a, b]$. By Theorem~\ref{thm:S-cmax-exists}, there exists $t^\star \in [a, b]$ with $C(t^\star) = C_{\max}$. Because $C$ is continuous on the compact interval $[a, b]$ it is uniformly continuous, so for every $\varepsilon > 0$ there exists $\delta > 0$ with $|t - s| < \delta \Rightarrow |C(t) - C(s)| < \varepsilon$. Since $\Delta(G_q) \to 0$, for all sufficiently large $q$ there exists $u_q \in G_q$ with $|u_q - t^\star| < \delta$, so
\[
C(u_q) > C(t^\star) - \varepsilon = C_{\max} - \varepsilon.
\]
Therefore $M_q = \max_{u \in G_q} C(u) \ge C(u_q) > C_{\max} - \varepsilon$. Also $M_q \le C_{\max}$ since $G_q \subset \Tcal$. For all sufficiently large $q$, $0 \le C_{\max} - M_q < \varepsilon$, and as $\varepsilon > 0$ was arbitrary, $M_q \to C_{\max}$. This holds for every continuous path, so convergence is almost sure.

For $L^1$, observe $|M_q-C_{\max}|\le|M_q|+|C_{\max}|\le2\sup_t|C(t)|$. If $\sup_t |C(t)|$ is integrable, dominated convergence yields $\E[|M_q-C_{\max}| \mid Y = y] \to 0$.
\end{proof}

\begin{corollary}[Distribution-function convergence for $C_{\max}$]
\label{cor:S-cmax-quantile}
Under the assumptions of Theorem~\ref{thm:S-grid-conv}, the distribution functions $F_q$ of $M_q$ converge to the distribution function $F$ of $C_{\max}$ at every continuity point of $F$.
\end{corollary}

\begin{proof}
Theorem~\ref{thm:S-grid-conv} gives $M_q \to C_{\max}$ almost surely, hence in distribution. Convergence in distribution is equivalent to convergence of distribution functions at all continuity points of the limit.
\end{proof}

\subsection{Maximum concentration inequalities}
\label{sec:S-conc-ineq}

We give relevant concentration inequalities here regarding the maxima of Gaussian distributions in this Section. We first document an upper tail bound and a lower tail bound, which together yield a two-sided interval result. 

\subsubsection*{Elementary finite-grid bounds}

\begin{theorem}[Elementary finite-grid upper-tail bound]
\label{thm:S-elem-upper}
Let $X = (X_1, \ldots, X_q)^\top \sim \Normal_q(\mu_X, \Sigma)$, $\mu_X = (\mu_{X,1}, \ldots, \mu_{X,q})^\top$, $M_q = \max_j X_j$, $m_q = \max_j \mu_{X,j}$, and $\sigma_q^2 = \max_j \Sigma_{jj}>0$. Then for every $r \ge 0$,
\[
\Prob(M_q \ge m_q + r) \le q \exp\!\left(-\frac{r^2}{2 \sigma_q^2}\right).
\]

\end{theorem}

\begin{proof}
For $r=0$, the bound is trivial because $q\ge1$. Suppose $r>0$.
The event $\{M_q \ge m_q + r\}$ implies $X_j \ge m_q + r$ for some $j$, which in turn implies $X_j - \mu_{X,j} \ge r$ since $\mu_{X,j} \le m_q$. By the union bound,
\[
\Prob(M_q \ge m_q + r) \le \sum_{j=1}^q \Prob(X_j - \mu_{X,j} \ge r).
\]
If $\Sigma_{jj}=0$, then $X_j-\mu_{X,j}=0$ almost surely, so $\Prob(X_j-\mu_{X,j}\ge r)=0$. For each coordinate with $\Sigma_{jj}>0$, the Gaussian tail bound gives
\[
\Prob(X_j - \mu_{X,j} \ge r) \le \exp\!\left(-\frac{r^2}{2 \Sigma_{jj}}\right) \le \exp\!\left(-\frac{r^2}{2 \sigma_q^2}\right).
\]
Summing over $j$ yields the claim.
\end{proof}

\begin{theorem}[Elementary finite-grid lower-tail bound]
\label{thm:S-elem-lower}
Under the assumptions of Theorem~\ref{thm:S-elem-upper}, for every $r \ge 0$,
\[
\Prob(M_q \le m_q - r) \le \exp\!\left(-\frac{r^2}{2 \sigma_q^2}\right).
\]
\end{theorem}

\begin{proof}
For $r=0$, the bound is trivial. Suppose $r>0$.
Pick any $j^\star \in \argmax_j \mu_{X,j}$, so $\mu_{X,j^\star} = m_q$. The event $\{M_q \le m_q - r\}$ implies $X_{j^\star} \le m_q - r = \mu_{X,j^\star} - r$. If $\Sigma_{j^\star j^\star}=0$, then $X_{j^\star}=m_q$ almost surely, so $\Prob(M_q\le m_q-r)=0$. Otherwise, $X_{j^\star}-\mu_{X,j^\star}\sim\Normal(0,\Sigma_{j^\star j^\star})$, and the Gaussian lower-tail bound gives
\[
\Prob(X_{j^\star} - \mu_{X,j^\star} \le -r) \le \exp\!\left(-\frac{r^2}{2 \Sigma_{j^\star j^\star}}\right) \le \exp\!\left(-\frac{r^2}{2 \sigma_q^2}\right).
\]
\end{proof}

\begin{corollary}[Finite-grid $C_{\max}$ interval]
\label{cor:S-elem-cmax}
Under the assumptions of Theorem~\ref{thm:S-elem-upper}, let $\alpha\in(0,1)$. With probability at least $1 - \alpha$, $M_q \le m_q + \sigma_q \sqrt{2 \log(q / \alpha)}$. With probability at least $1 - \alpha$, $M_q \ge m_q - \sigma_q \sqrt{2 \log(1 / \alpha)}$. Hence with probability at least $1 - 2\alpha$,
\[
m_q - \sigma_q \sqrt{2 \log(1/\alpha)} \le M_q \le m_q + \sigma_q \sqrt{2 \log(q/\alpha)}.
\]
\end{corollary}

\begin{proof}
Set $r=\sigma_q\sqrt{2\log(q/\alpha)}$ in Theorem~\ref{thm:S-elem-upper} to obtain
\[
q\exp\{-r^2/(2\sigma_q^2)\}=\alpha.
\]
Similarly, setting $r=\sigma_q\sqrt{2\log(1/\alpha)}$ in Theorem~\ref{thm:S-elem-lower} gives
\[
\exp\{-r^2/(2\sigma_q^2)\}=\alpha.
\]
A union bound yields the two-sided statement.
\end{proof}

\subsubsection*{Gaussian Lipschitz concentration}
We now give related Lipschitz concentration results, which can lead to sharper bounds. 
\begin{lemma}[Gaussian Lipschitz concentration]
\label{lem:S-glipschitz}
Let $Z \sim \Normal_d(0, I_d)$ and let $f:\R^d \to \R$ be $L$-Lipschitz with $L>0$: $|f(z) - f(z')| \le L\|z - z'\|_2$. Then for every $r \ge 0$,
\[
\Prob\{f(Z) - \E f(Z) \ge r\} \le \exp\!\left(-\frac{r^2}{2 L^2}\right),
\qquad
\Prob\{f(Z) - \E f(Z) \le -r\} \le \exp\!\left(-\frac{r^2}{2 L^2}\right).
\]
\end{lemma}

\begin{proof}
This is the classical Gaussian-Lipschitz concentration inequality, due in this form to \citet{borell1975brunn} and \citet{sudakov1974extremal}; a textbook presentation appears in \citet{boucheron2013concentration} as well as in \citet{wainwright2019highdimensional}.
\end{proof}
We now apply this bound to get a sharper finite-grid concentration result of the maximum. 
\begin{theorem}[Sharper finite-grid maximum concentration]
\label{thm:S-sharp-finitegrid}
Let $X = (X_1, \ldots, X_q)^\top \sim \Normal_q(\mu_X, \Sigma)$ and $M_q = \max_j X_j$. Let $\sigma_q^2 = \max_j \Sigma_{jj}>0$. Then for every $r \ge 0$,
\[
\Prob(M_q - \E M_q \ge r) \le \exp\!\left(-\frac{r^2}{2 \sigma_q^2}\right),
\quad
\Prob(M_q - \E M_q \le -r) \le \exp\!\left(-\frac{r^2}{2 \sigma_q^2}\right).
\]
\end{theorem}

\begin{proof}
Let $\Gamma$ be a matrix with $\Gamma \Gamma^\top = \Sigma$ (e.g., a Cholesky factor), and let $Z \sim \Normal_q(0, I_q)$. Then $X \stackrel{d}{=} \mu_X + \Gamma Z$. Write $\Gamma_j^\top$ for the $j$-th row of $\Gamma$, and define
\[
f(z) = \max_{j=1, \ldots, q}\{\mu_{X,j} + \Gamma_j^\top z\}, \qquad z \in \R^q.
\]
Then $M_q \stackrel{d}{=} f(Z)$. For any $z, z'$,
\[
|f(z) - f(z')| \le \max_j |\Gamma_j^\top (z - z')| \le \max_j \|\Gamma_j\|_2\, \|z - z'\|_2.
\]
Since $\|\Gamma_j\|_2^2 = \Sigma_{jj}$, $f$ is Lipschitz with constant $\sigma_q = \sqrt{\max_j \Sigma_{jj}}$. Lemma~\ref{lem:S-glipschitz}, applied with Lipschitz constant $L = \sigma_q$, yields the bound.
\end{proof}

We now extend to the continuous-time case. 

\begin{theorem}[Continuous-time maximum concentration]
\label{thm:S-cts-conc}
Let $X(t), t \in \Tcal$, be a centered separable Gaussian process with almost surely continuous sample paths on compact $\Tcal = [a, b]$. Let $\mu:\Tcal \to \R$ be deterministic and continuous, $C(t) = \mu(t) + X(t)$, and $M = \sup_{t \in \Tcal} C(t)$. Let $\sigma_\Tcal^2 = \sup_{t \in \Tcal} \Var\{X(t)\}>0$. Assume $\E M < \infty$. Then for every $r \ge 0$,
\[
\Prob(M - \E M \ge r) \le \exp\!\left(-\frac{r^2}{2 \sigma_\Tcal^2}\right),
\quad
\Prob(M - \E M \le -r) \le \exp\!\left(-\frac{r^2}{2 \sigma_\Tcal^2}\right).
\]
\end{theorem}

\begin{proof}
Let $G_1 \subset G_2 \subset \cdots$ be an increasing sequence of finite grids containing a point of positive variance, with $\bigcup_q G_q$ dense in $\Tcal$ and $\Delta(G_q) \to 0$. Define $M_q = \max_{t \in G_q} C(t)$. By Theorem~\ref{thm:S-grid-conv}, $M_q \to M$ almost surely; because the grids are nested, $M_q \uparrow M$. Since $M_1$ has finite expectation and $M$ has finite expectation, monotone convergence applied to $M_q - M_1$ gives $\E M_q \to \E M$.

For each finite grid $G_q$, $X(G_q) = (X(t))_{t \in G_q}$ is multivariate Gaussian, and Theorem~\ref{thm:S-sharp-finitegrid} yields
\[
\Prob(M_q - \E M_q \ge r) \le \exp\!\left(-\frac{r^2}{2 \sigma_q^2}\right),
\qquad \sigma_q^2 = \max_{t \in G_q} \Var\{X(t)\}.
\]
Because $\sigma_q^2 \le \sigma_\Tcal^2$, we have $\Prob(M_q - \E M_q \ge r) \le \exp(-r^2 / (2 \sigma_\Tcal^2))$. Now $M_q - \E M_q \to M - \E M$ in distribution by the almost-sure convergence combined with $\E M_q \to \E M$. Taking limits at continuity points gives the upper-tail bound there. For any $r>0$, choose continuity points $r_n\uparrow r$ and use $\{M-\E M\ge r_n\}\downarrow\{M-\E M\ge r\}$ to obtain the bound at $r$; the case $r=0$ is immediate. Apply the same argument to the lower-tail inequalities from Theorem~\ref{thm:S-sharp-finitegrid}.
\end{proof}

\subsubsection*{Posterior concentration}
We now apply the above results to prove Theorem~\ref{thm:cmax-conc} in the main text. 
\begin{proof}[Proof of Theorem~\ref{thm:cmax-conc}]
Conditional on $Y=y$, the centered process $X(t)=C(t)-\mu(t)$ satisfies Theorem~\ref{thm:S-cts-conc}, with $M=C_{\max}$ and $\sigma_\Tcal=\sigma_\star$. Applying its two one-sided bounds and taking a union bound gives the first inequality. Setting $r=\sigma_\star\sqrt{2\log(2/\alpha)}$ gives the second.
\end{proof}

\clearpage

\section{Supporting illustrations and case studies}
\label{sec:support-cases}

\subsection{Sparse-sampling simulation}
\label{sec:support-sparse}

Figure~\ref{fig:sparse} illustrates the fitted concentration--time curve for a representative patient in the sparse oral scenario S4 of Section~\ref{sec:sim}.

\begin{figure}[!htbp]
\begingroup
\spacingset{1}
\centering
\includegraphics[width=\linewidth]{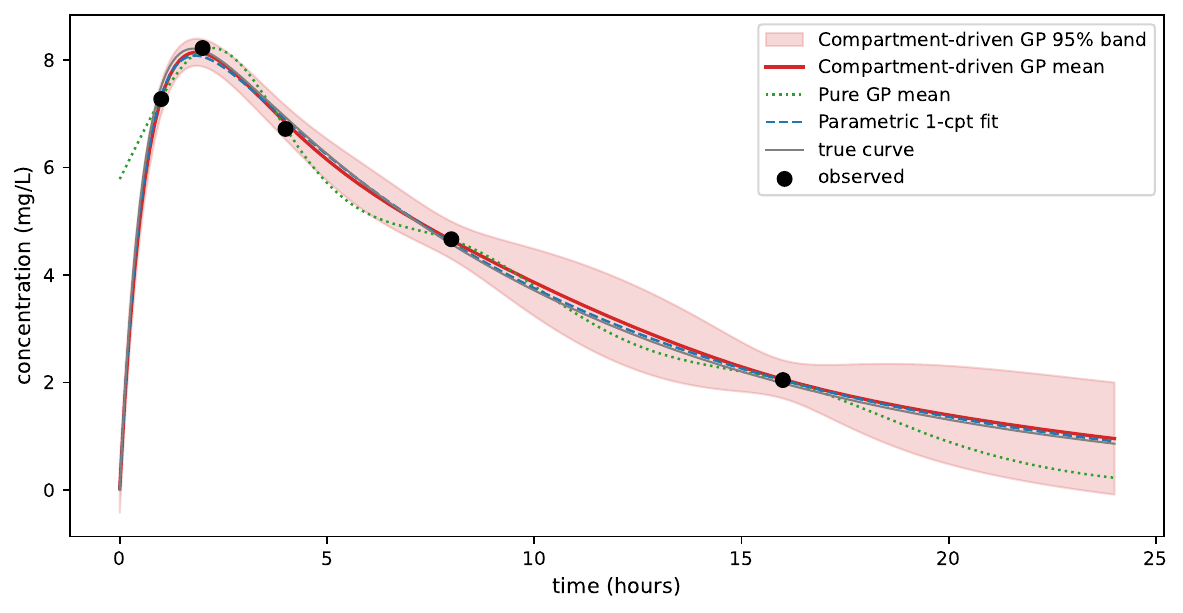}

\caption{Sparse one-compartment oral scenario (S4). For a representative patient with sampling times $t \in \{1, 2, 4, 8, 16\}$\,hr: observed points; true curve (gray); posterior means of the compartment-driven GP (solid), pure GP (dotted), and parametric one-compartment fit (dashed); $95\%$ credible band of the compartment-driven GP shaded. The posterior band widens in the terminal phase beyond the last sample at $16$\,hr. \label{fig:sparse}}
\endgroup
\end{figure}

\subsection{Warfarin concentration--time fits}
\label{sec:support-warfarin}

Figure~\ref{fig:warfarin-patients} shows posterior concentration--time fits for two representative patients from the warfarin analysis in Section~\ref{sec:warfarin}.

\begin{figure}[!htbp]
\begingroup
\spacingset{1}
\centering
\includegraphics[width=\linewidth]{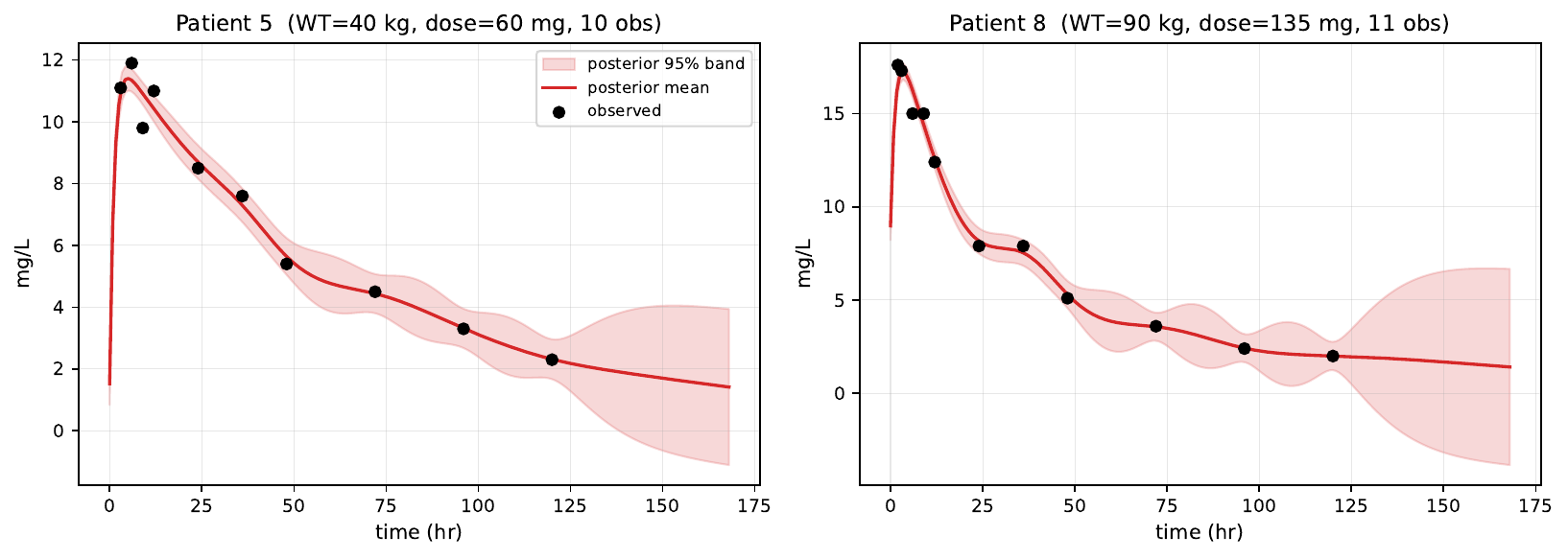}

\caption{Warfarin case study: posterior fits for two representative patients. Each panel shows observed plasma concentrations (black dots), the posterior mean (red solid), and the posterior 95\% pointwise band (red shaded).\label{fig:warfarin-patients}} 
\endgroup
\end{figure}

\subsection{Chemotherapy concentration--time fits}
\label{sec:support-chemo}

Figures~\ref{fig:5fu-patients} and~\ref{fig:cape-patients} show the cohort-mean concentration--time fits underlying the chemotherapy AUC comparisons in Section~\ref{sec:chemo}.

\begin{figure}[!htbp]
\begingroup
\spacingset{1}
\centering
\includegraphics[width=\linewidth]{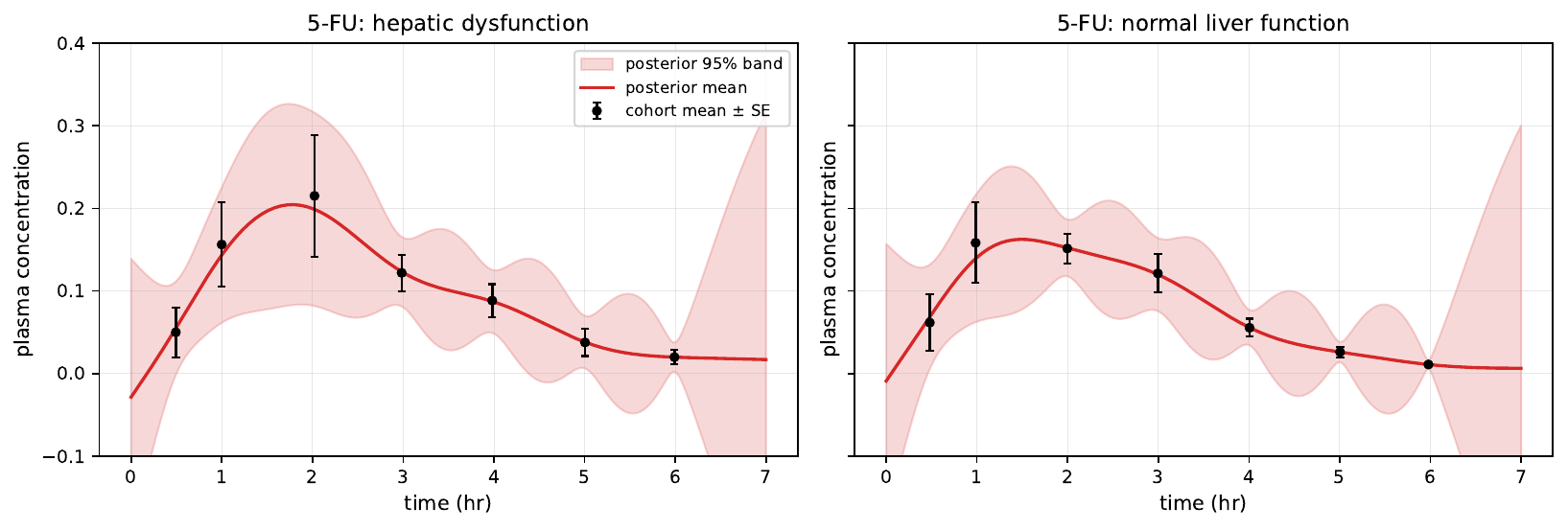}

\caption{5-FU case study: compartment-driven GP fits to the cohort-mean concentration--time curves digitized from the source clinical-pharmacology study. Left panel: hepatic dysfunction cohort. Right panel: normal liver function cohort. Each panel shows observed cohort-mean concentrations with per-timepoint standard errors (black with bars), the posterior mean curve (red solid), and the posterior 95\% pointwise band (red shaded).\label{fig:5fu-patients}}
\endgroup
\end{figure}

\begin{figure}[!htbp]
\begingroup
\spacingset{1}
\centering
\includegraphics[width=\linewidth]{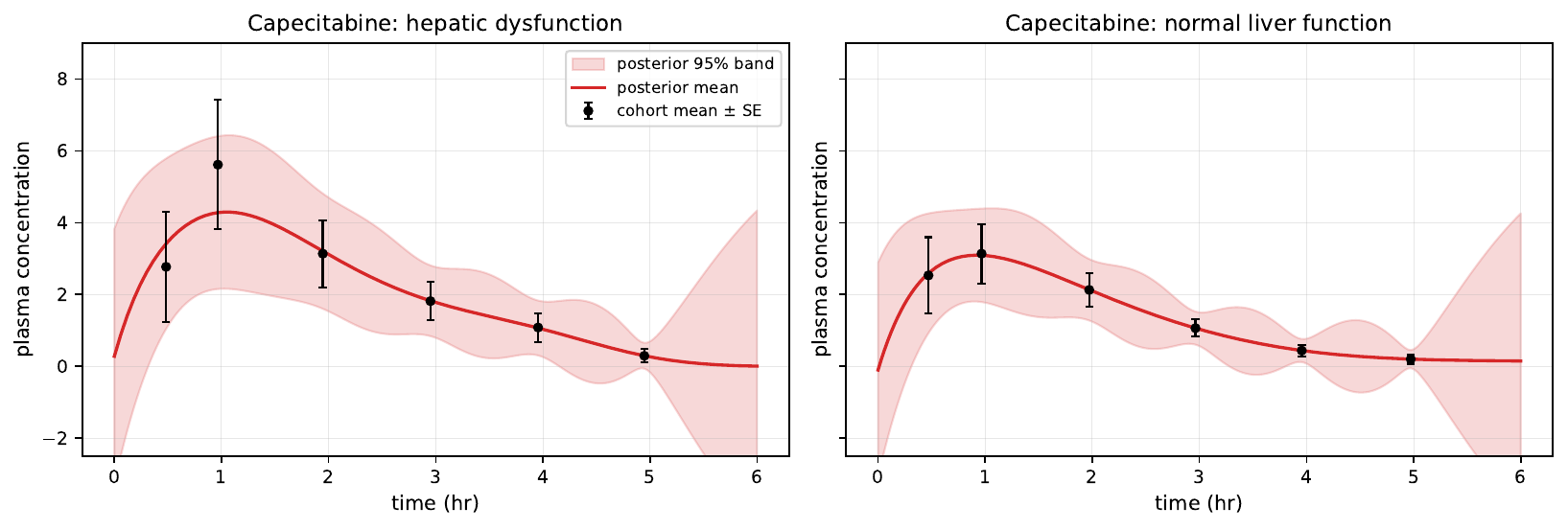}

\caption{Capecitabine case study: compartment-driven GP fits to the cohort-mean concentration--time curves digitized from the source clinical-pharmacology study. Left panel: hepatic dysfunction cohort. Right panel: normal liver function cohort. Each panel shows observed cohort-mean concentrations with per-timepoint standard errors (black with bars), the posterior mean curve (red solid), and the posterior 95\% pointwise band (red shaded).\label{fig:cape-patients} 
}
\endgroup
\end{figure}

\subsection{Theophylline case study}
\label{sec:theo}

This section reports a second individualized case study, analyzed with the same model and the settings tabulated in Section~\ref{sec:hp-app}. The theophylline dataset contains plasma concentrations of theophylline measured at 11 timepoints in each of 12 healthy adult subjects after a single oral dose of approximately 3--6\,mg/kg. This is a widely used benchmark dataset for pharmacokinetic modeling. Textbook adult population values $\kappa_a = 1.5\,\text{hr}^{-1}$ and $\kappa_e = 0.08\,\text{hr}^{-1}$ (half-life $\approx 9\,$hr) were used to specify the compartment model prior mean, with the dose entering on the concentration scale as $c_{0,i} = D_i / (v_{\text{kg}} W_i)$ and $v_{\text{kg}} = 0.5\,\text{L/kg}$.  

Figure~\ref{fig:theo-patients} shows posterior fits for two representative patients. 

\begin{figure}[!htbp]
\begingroup
\spacingset{1}
\centering
\includegraphics[width=\linewidth]{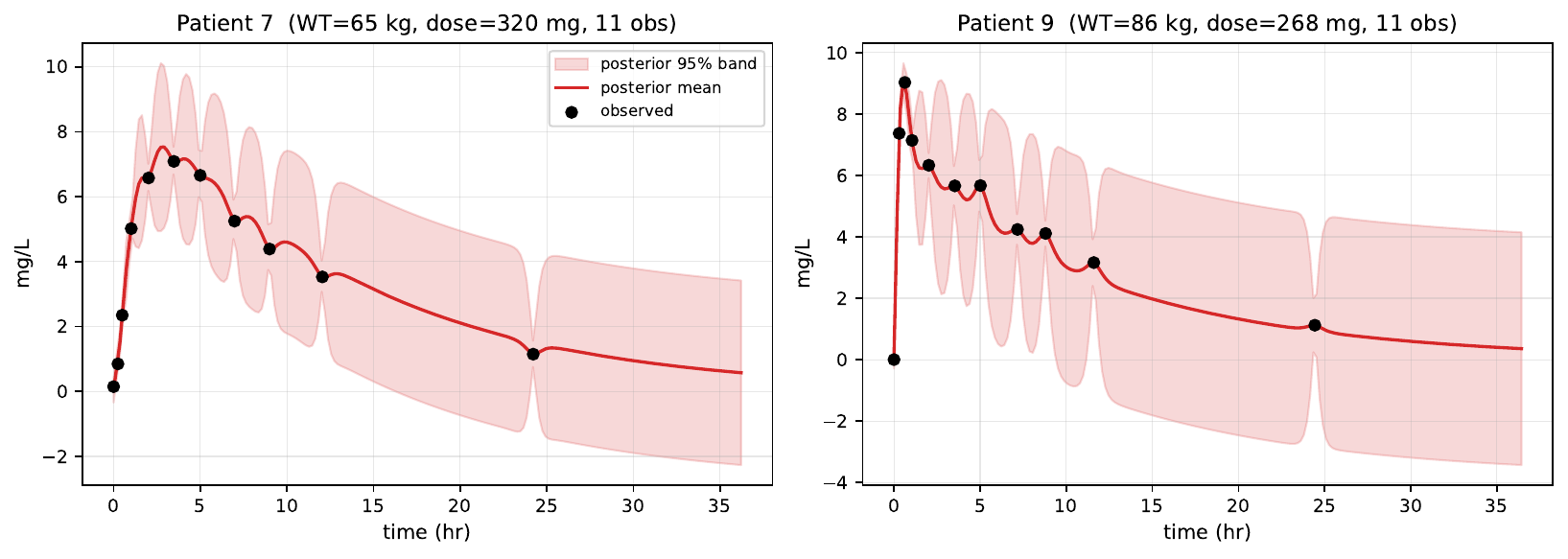}

\caption{Theophylline case study: posterior fits for representative patients, with the same legend convention as Figure~\ref{fig:warfarin-patients}.\label{fig:theo-patients}}
\endgroup
\end{figure}

Patient-level posterior 95\% intervals for $\AUC$ over $[t_{\text{dose}}, t_{\text{dose}} + 30\,\text{hr}]$ and for $C_{\max}$ are shown in Figure~\ref{fig:theo-intervals}. The cohort displays the expected one-compartment oral profile, with posterior median exposures spanning 81--131\,mg$\cdot$hr/L for $\AUC$ and 7.5--11.9\,mg/L for $C_{\max}$.

\begin{figure}[!htbp]
\begingroup
\spacingset{1}
\centering
\includegraphics[width=\linewidth]{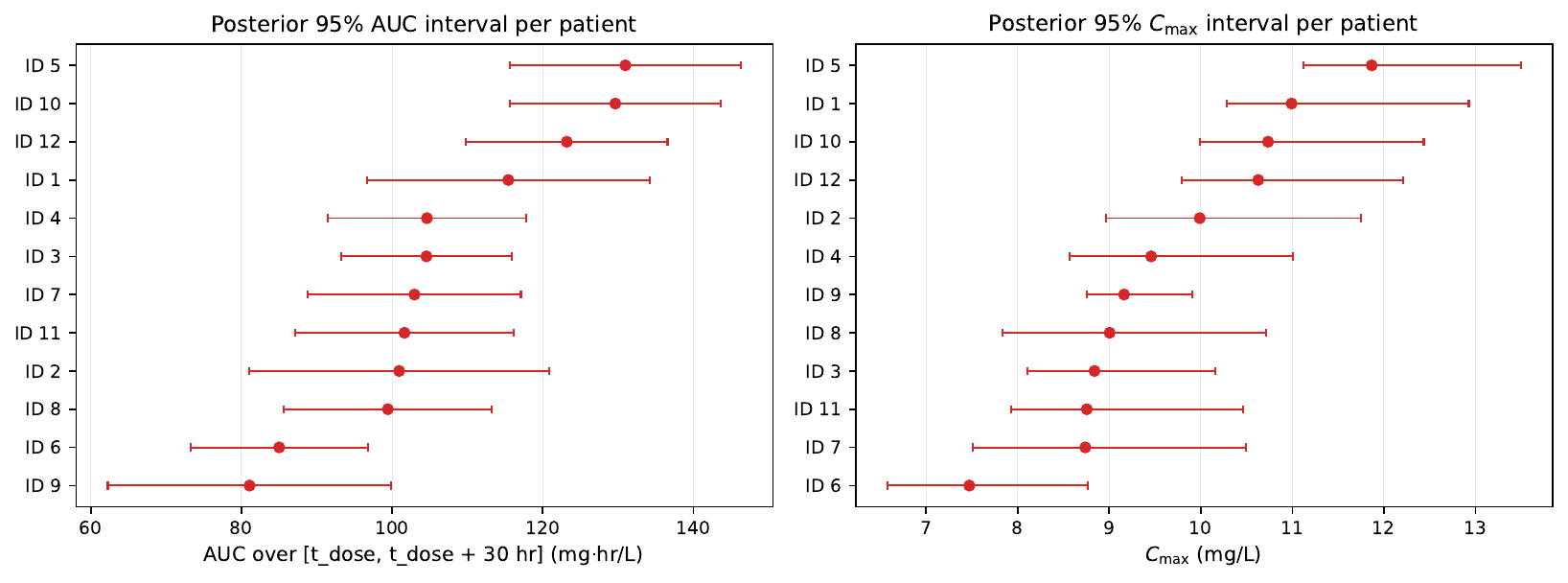}

\caption{Theophylline case study: posterior 95\% credible intervals per patient for $\AUC$ over $[t_{\text{dose}}, t_{\text{dose}} + 30\,\text{hr}]$ (left) and for $C_{\max}$ (right). Patients are sorted by posterior median.\label{fig:theo-intervals}}
\endgroup
\end{figure}

Aggregate posterior summaries are reported in Table~\ref{tab:theo-summary}.

\begin{table}
\caption{Theophylline case study: cohort summary of patient-level posterior exposures and credible interval widths across the 12 patients. \label{tab:theo-summary}}
\begin{center}
\begin{tabular}{lccc}
\toprule
Quantity & Minimum & Median & Maximum \\
\midrule
Posterior $\AUC$ median (mg$\cdot$hr/L) & 81 & 104 & 131 \\
Posterior $C_{\max}$ median (mg/L) & 7.5 & 9.3 & 11.9 \\
$\AUC$ 95\% interval width (mg$\cdot$hr/L) & 23 & 28 & 40 \\
$C_{\max}$ 95\% interval width (mg/L) & 1.15 & 2.44 & 2.98 \\
\bottomrule
\end{tabular}
\end{center}
\end{table}

\bibliographystyle{chicago}

\bibliography{pkgp}
\end{document}